\documentclass[12pt]{article}
\usepackage[margin=1in, centering]{geometry}
\usepackage{mathtools}
\usepackage{amssymb}
\usepackage{amsthm}
\usepackage{xcolor}
\usepackage{esvect}
\usepackage{bbm}
\usepackage{times}

\usepackage[pdfstartview=FitH,pdfpagemode=UseNone,colorlinks=true,citecolor=blue,linkcolor=blue]{hyperref}

\usepackage[capitalise]{cleveref}

\newtheorem{thm}{Theorem}[section]
\newtheorem{prop}[thm]{Proposition}
\newtheorem{lem}[thm]{Lemma}
\newtheorem{cor}[thm]{Corollary}
\newtheorem{claim}[thm]{Claim}

\newtheorem{definition}[thm]{Definition}

\newcommand{\R}{\mathbb{R}}  

\newcommand{\N}{\mathbb{N}}  
     
\DeclareMathOperator*{\E}{\mathbb{E}}  
  
\DeclareMathOperator*{\PP}{\mathbb{P}}  
\renewcommand{\P}{\PP} 
\renewcommand{\Pr}{\PP}  
\newcommand{\1}{ \mathbbm{1}}  
\newcommand{\eps}{ \varepsilon}  

\newcommand{\set}[2]{\{#1 \: : \: #2\}}
\newcommand{\ip}[2]{\left\langle#1 ,#2 \right\rangle}

\newcommand{\paren}[1]{\left(#1 \right)}

\newcommand{\ind}[1]{\1 \! \! \paren{#1}}
\newcommand{\prob}[1]{\P \!  \paren{#1}}

\newcommand{\mnote}[1]{{\color{blue}{Raghu: #1}}}
\newcommand{\hvar}{\textnormal{HyperVar}}
\newcommand{\ignore}[1]{{}}

\newcommand{\sign}{\mathsf{sign}}

\begin{document}

\title{Random restrictions and PRGs for PTFs in Gaussian Space}

\author{Zander Kelley\footnote{Department of Computer Science, University of Illinois at Urbana-Champaign. Supported by NSF grants CCF-1755921 and CCF-1814788. Email: \texttt{awk2@illinois.edu}}
,\;\; Raghu Meka\footnote{Department of Computer Science, University of California, Los Angeles. Supported by NSF Career  Award 1553605 and NSF AF 2007682. Email: \tt{raghum@cs.ucla.edu}} }
\date{}
\maketitle

\begin{abstract}
     A polynomial threshold function (PTF) $f:\mathbb{R}^n \rightarrow \mathbb{R}$ is a function of the form $f(x) = \mathsf{sign}(p(x))$ where $p$ is a polynomial of degree at most $d$. PTFs are a classical and well-studied complexity class with applications across complexity theory, learning theory, approximation theory, quantum complexity and more. We address the question of designing pseudorandom generators (PRGs) for polynomial threshold functions (PTFs) in the gaussian space: design a PRG that takes a seed of few bits of randomness and outputs a $n$-dimensional vector whose distribution is indistinguishable from a standard multivariate gaussian by a degree $d$ PTF. 
     
     Our main result is a PRG that takes a seed of $d^{O(1)}\log ( n / \varepsilon)\log(1/\varepsilon)/\varepsilon^2$ random bits with output that cannot be distinguished from $n$-dimensional gaussian distribution with advantage better than $\varepsilon$ by degree $d$ PTFs. The best previous generator due to O'Donnell, Servedio, and Tan (STOC'20) had a quasi-polynomial dependence (i.e.,  seedlength of $d^{O(\log d)}$) in the degree $d$. Along the way we prove a few nearly-tight  structural properties of \emph{restrictions} of PTFs that may be of independent interest.
\end{abstract}
\newpage

\setcounter{page}{1}

\section{Introduction}
Polynomial threshold functions (PTFs) are a classical and well-studied class of functions with several applications in complexity theory, learning theory, theory of approximation, and more. Here we study the question of designing \emph{pseudorandom generators} (PRGs) that fool test functions that are PTFs. We first start with some standard definitions. Let $\sign:\R \rightarrow \R$ be defined as $\sign(z) = 1$ if $z \geq 0$ and $0$ otherwise.

\begin{definition}
For an integer $d > 0$, a degree $d$ PTF $f:\R^n \rightarrow \{0,1\}$ is a function of the form $f(x) = \sign(p(x))$, where $p:\R^n \rightarrow \R$ is a polynomial of degree at most $d$.
\end{definition}

Our goal is to design a PRG that takes few random bits and outputs a high-dimensional vector whose distribution is indistinguishable from a standard multivariate gaussian by any low-degree PTF. Specifically:

\begin{definition} \label{def:prg}
A function $G:\{0,1\}^r \rightarrow \R^n$ is a pseudorandom generator for degree $d$ PTFs with error $\varepsilon$ if for every degree at most $d$ PTF $f:\R^n \rightarrow \{0,1\}$, 
$$\left|\Pr_{y \in_u \{0,1\}^r}\left( f(G(y)) = 1\right) - \Pr_{x \sim N(0,1)^n} \left(f(x) = 1 \right) \right| \leq \varepsilon.$$

We call $r$ the seedlength of the generator and say $G$ $\varepsilon$-fools degree $d$ PTFs with respect to the gaussian distribution 
\footnote{Here, and henceforth, $y \in_u S$ denotes a uniformly random element from a multi-set $S$, and $N(0,1)$ represents the standard univariate gaussian distribution of variance $1$.}.
We say $G$ is explicit if its output can be computed in time polynomial in $n$.
\end{definition}

Of particular interest is the \emph{boolean case} where the target distribution is not gaussian but the uniform distribution on the hypercube $\{+1,-1\}^n$. The gaussian case is interesting by itself both from a complexity-theoretic view as well as a geometric one. For instance, a PRG as above can be used to get deterministic algorithms for approximating the gaussian volumes of polynomial surfaces. Further, the gaussian case is a necessary stepping-stone to obtaining PRGs in the Boolean case: a PRG for the latter implies a PRG for the gaussian case. Achieving similar parameters as we do for the boolean case would be a significant achievement: we do not even have non-trivial correlation lower bounds for NP\footnote{A PRG would at the very least imply correlation lower bounds against a function in NP.} against PTFs of degree $\omega(\log n)$ over the hypercube, a longstanding bottleneck in circuit complexity. 


Over the last several years, the question of designing PRGs for PTFs has received much attention. Non-explicitly (i.e., the generator is not necessarily efficiently computable), by the probabilistic method, it is known that there exists PRGs that $\varepsilon$-fool degree $d$ PTFs with seed-length is $O(d \log n  + \log(1/\varepsilon))$. Meka and Zuckerman \cite{MZ13} gave the first non-trivial PRG for bounded degree PTFs with a seedlength of $d^{O(d)}\log(n)/\varepsilon^2$ for the boolean and gaussian cases. Independent of \cite{MZ13}, \cite{DKN10} showed that bounded independence fools degree-$2$ PTFs leading to seedlength $O(\log(n)/\varepsilon^2)$. Since then, there have been several other works that make progress on the gaussian case \cite{Kane11a, Kane11b, Kane12, Kane14, Kane15}. The seedlength in all of these works had an exponential dependence on the degree $d$ of the PTF. In particular, until recently no non-trivial PRGs (i.e., seedlength $o(n)$) were known for PTFs of degree $\omega(\log n)$. In a remarkable recent work, O'Donnell, Servedio, and Tan \cite{OST20} got around this exponential dependence on the degree $d$, achieving a seedlength of  $(d/\eps)^{O(\log d)} \log(n)$. Our work builds on their work (which in turn builds on a framework of \cite{Kane11a}).

\subsection{Main Results}

Our main result is a PRG that $\epsilon$-fools $n$-variate degree-$d$ PTFs with seed-length  $(d/\epsilon)^{O(1)} \log (n)$:

\begin{thm}[PRG for PTFs]\label{th:mainintro}
There exist constants $c,C$ such that for all $\epsilon > 0$ and $d \geq 1$, there exists an explicit PRG that $\varepsilon$-fools $n$-variate degree $d$ PTFs with respect to the gaussian distribution with seedlength $r(n,d,\epsilon) = C d^{c} \log(n/\eps) \log(1/\eps) /\eps^2$. 
\end{thm}

As remarked above, this is the first result with polynomial dependence on the degree for fooling PTFs against any distribution and gives the first non-trivial PRGs against PTFs of degree $n^{\Omega(1)}$. Previously, we could only handle degree at most $2^{O(\sqrt{\log n})}$.

Towards proving the above result, we develop several structural results on PTFs in the gaussian space that might be of independent interest. We expand on these later on. Briefly:
\begin{itemize}
\item We show that the derivatives of a low-degree polynomial $p$,
taken at a random point $x \sim N(0,1)^n$,
are likely to have magnitudes $\| \nabla^k p(x) \|$ which grow slowly as $k$ increases. 
\item We apply this fact to the study of random ``gaussian restrictions" of a polynomial $p$,
$$ p_{x,\lambda}(y) := p\left( \sqrt{1-\lambda}x + \sqrt{\lambda} y\right), $$
and conclude that for small enough $\lambda$, with high probability over $x \sim N(0,1)^n$, $p_{x,\lambda}(Y)$ becomes highly concentrated around its mean value $\mu$ when $Y \sim N(0,1)^n$, as quantified by a bound on the higher-moments $\E (p_{x,\lambda}(Y) - \mu)^R $.
\item As this concentration result relies only on moment bounds, it extends easily to pseudorandom distributions $Y$ over $\R^n$ which are $k$-moment-matching with $N(0,1)^n$, when $k \geq R \cdot \textnormal{deg}(p)$. 
\end{itemize}
Note that the magnitudes of the derivatives $\nabla^k p_{x,\lambda}(0)$ (with respect to $x$) are the same as the magnitudes of the degree-$k$ coefficients of $p_{x,\lambda}(y)$ (as a polynomial in $y$), up to a scaling factor of roughly $\lambda^{k/2}$. However, to obtain the moment bound, we must translate to the basis of Hermite polynomials and bound the degree-$k$ coefficients with respect to this basis (rather than the standard basis). 
In contrast with our work, \cite{OST20} derive coefficient-size bounds for the Hermite basis directly and work with it exclusively. However, there are some significant advantages in having the flexibility to work also within the standard basis which will become relevant later -- mainly they are due to the fact that standard basis representations (or equivalently: derivatives) behave nicely under the scaling operator $p(t) \mapsto p(\gamma t)$. The Hermite-basis representation behaves poorly under scaling\footnote{In contrast, the Hermite basis representation behaves nicely under the \emph{noise operator}, $p(t) \mapsto \E_{x \sim N(0,1)^n} p(\sqrt{1-\lambda} x + \sqrt{\lambda} t)$.}. 

For an arbitrary fixed polynomial $p(t)$, a bound on the coefficient-sizes in one basis translates only to a fairly crude bound in the other basis\footnote{This is especially true in the current setting where we must control the \emph{relative} sizes of the magnitudes of coefficients at degree $k$ vs.\ $k+1$.}. Therefore, we come to the following rather technical contribution of our work which we would like to highlight: we find that, although it is rather painful to convert between bases while studying an arbitrary fixed polynomial, it is actually quite possible to do so when studying certain \emph{average-case} behaviors of polynomials; for instance, to study the typical behavior of $p(x)$ in the neighborhood around a random point $x \sim N(0,1)^n$, or the typical moments of $p(\sqrt{1-\lambda}x + \sqrt{\lambda} Y)$, it is possible to pass freely between either polynomial basis, and we develop some simple tools for doing so. These tools appear to be new (at least with respect to the body of works on PTFs) and it seems likely that they could be helpful in future works.

Besides these structural results and technical contributions, we also manage to introduce some substantial simplifications to the analysis of the main PRG as compared to \cite{OST20}. This is in part due to the flexibility we have to measure the \emph{well-behavedness} of a polynomial $p$ in the neighborhood around a point $x$ directly via the derivatives at $x$, rather than indirectly by taking several Hermite expansions of $p$ and other auxiliary polynomials (cf.~\emph{horizantal, diagonal mollifier checks} in \cite{OST20}). We will expand on this in \cref{sec:overview} when discussing our analysis, but we briefly summarize a few key points here.
\begin{itemize}
    \item Following \cite{Kane11a} and \cite{OST20}, the pseudorandom construction we analyze is of the form $Z := \frac{1}{\sqrt{L}}\sum_{i=1}^L Y_{i}$, where each $Y_i$ is a $k$-moment-matching gaussian. This can be thought of as the gaussian analogue of the boolean construction from \cite{MZ13}, which pseudorandomly partitions the $n$ input bits into $L$ buckets, and then assigns the bits in each bucket using $k$-wise independence. This construction and its variants are by now the most widely-applied pseudorandom tool for fooling various ``geometric" function classes including linear threshold functions and their generalizations (such as PTFs and intersections of halfspaces). 
    \item A tempting first idea for analyzing $Z$ is to apply a hybrid argument -- this seems promising in light of the fact that for a low-degree polynomial, we know that $p\left(\sqrt{1-\frac{1}{L}} x + \sqrt{\frac{1}{L}}Y_i\right)$ should be highly-concentrated around its mean for typical $x$. However, this naive idea fails quantitatively: The probability that we have good behavior at $x$ is in general not smaller than $\sqrt{1/L}$, so we cannot afford a union-bound over $L$ events as required by the standard hybrid argument. Remarkably in \cite{Kane11a}, Kane shows how to address this obstacle with a clever sandwiching argument which in some sense mimics the hybrid argument but manages to pay for the error caused by ``bad points" $x$ only once rather than $L$ times.
    \item However, one drawback of Kane's analysis is that its implementation is highly elaborate. After the framework was extended by \cite{OST20} to break the $\log(n)$-degree barrier, the complexity only increased and the details of the argument became only more specialized and technical\footnote{Refer to \cite{OSTK21}, which fills in several details absent in \cite{OST20}, to see the full scope of the argument.}. Given the wide applicability of the aforementioned pseudorandom construction and its variants,
    it would be highly desirable to have a lean and more transparent analysis which might better serve as a flexible starting point for future adaptations. We propose that in this work, we do obtain such an analysis.  
    
\end{itemize}



\paragraph{PTFs simplify under restrictions.}
As a byproduct of our analysis, we obtain a structural result on PTFs that is similar in spirit to the celebrated \emph{switching lemmas} that show that certain classes of functions simplify significantly under random restrictions. Switching lemmas and random restrictions are a cornerstone in complexity theory, and are one of the main methods we have for proving lower bounds. We prove analogous results with nearly optimal parameters for the important class of PTFs in the continuous space. 

In the \emph{boolean case}, i.e., when studying distributions on the hypercube $\{+1,-1\}^n$, a \emph{restriction} is a partial assignment of the form $\rho \in \{+1,-1,*\}^n$ with the understanding that the $*$-variables are left free. Typically, restrictions $\rho$ as above are parametrized by some $\lambda > 0$, the fraction of $*$'s.


Here, we study analogues of the above results in the continuous world, where the inputs are coming from the standard gaussian distribution. The first question however is what should the analogue of random restrictions be in the continuous space? As it turns out, adopting the usual interpretation (where some coordinates are fixed and some are free) is not a natural one to study in the continuous space especially for PTFs\footnote{One reason is that the class of PTFs is invariant under linear transformations, so it would be nice to have our notion of restrictions also have some symmetry under linear transformations.}. 

The answer comes from the work of \cite{Kane11a} (further developed in \cite{OST20}) who introduced the notion of a \emph{zoom} of a polynomial. To draw a clearer parallel with random restrictions, we term these \emph{gaussian restrictions}:

\begin{definition}
Given a function $p:\R^n \rightarrow \R$ and $x \in \R^n$, and a \emph{restriction} parameter $\lambda \in (0,1)$, let $p_{x,\lambda}:\R^n \rightarrow \R$ be\footnote{As the value of $\lambda$ will often be clear, we will often in fact just use $p_x$ for brevity.} the function $p_{x,\lambda}(y) = p(\sqrt{1-\lambda} x + \sqrt{\lambda} y)$. 
\end{definition}
Intuitively, we can view $p_{x,\lambda}$ as a restriction where $(1-\lambda)$-fraction of the \emph{variance} is already \emph{fixed}. (Note that for independent $x,y \sim N(0,1)^n$, $\sqrt{1-\lambda}x + \sqrt{\lambda} y$ is distributed as $N(0,1)^n$.)

We show that PTFs simplify significantly, i.e., become essentially constant, under \emph{gaussian restrictions} for $\lambda \ll 1/d^6$. 

\newpage


\begin{thm}\label{th:rest}
There is a constant $C > 0$ such that the following holds. For any $\delta,\varepsilon > 0$, if
\begin{itemize}
    \item $f:\R^n \rightarrow \{0,1\}$ is a PTF of  degree $d$, and
    \item $\lambda \leq C \frac{\delta^2}{d^6\log(1/\varepsilon)}$,
\end{itemize}
 then with probability at least $1-\delta$ over $x \sim N(0,1)^n$, the gaussian restriction of the PTF  ($f_{x,\lambda}$) is \emph{nearly fixed to a constant}: for some $b \in \{0,1\}$ we have 
$$\Pr_{y \sim N(0,1)^n}[f_{x,\lambda}(y) = b] > 1 - \varepsilon.$$
\end{thm}


The work of \cite{OST20} achieves a similar conclusion but when the restriction parameter is $\lambda = d^{-O(\log d)}$ as opposed to being polynomially small as above. This improved significantly on the work of \cite{Kane11a} that implicitly shows a similar claim for $\lambda = 2^{-O(d)}$. 
\ignore{
{ \color{green} Note that our bound is almost tight up to the constant in the polynomial dependence: one needs at least $\lambda = 1/d^2$. }
\mnote{Footnote on example?}
{\color{red} Zander: the univariate liner function p(x) = x should suffice. Actually, we should probably discuss optimally of the hypervariance reduction lemma itself rather than this cor.}\mnote{What do you mean? How would d come into the picture ... I am thinking delta is say 1/10. We should give examples for both.} {\color{red} Zander: I think you wrote this -- maybe you meant $\delta^2$?}\mnote{True. How about the dependence on d?} {\color{red} Zander: I'm not sure. I'm setting this sentence aside in green for now.}}

We remark that in a related line of work, \cite{ELY09,HKM14,DRST14, KKL17} study random restrictions of PTFs over the hypercube. Our focus here is on gaussian restrictions and obtaining stronger bounds quantitatively: these works had exponential dependence on the degree d.


\paragraph{Slow-growth of derivatives.} The analysis of the PRG (\cref{th:prgmain}) and the random restriction statement above (\cref{th:rest}) rely crucially on a claim about the magnitude of the derivatives of a  polynomial evaluated at random gaussian input which may itself be of independent interest (and can be stated in a self-contained way). 

For a function $p:\R^n \rightarrow \R$, let $\|\nabla^k p(x)\|^2$ denote the sum of squares of all partial derivatives of $p$ of order $k$ at $x$. That is, $\|\nabla^k p(x)\|$ is the Frobenius norm of the tensor of $k$'th order partial derivatives of $p$.  We show that for any degree $d$ polynomial $p$, the Frobenius-norm of the $k$'th order derivatives are comparable to the $(k-1)$'th order derivatives on a random gaussian input with high probability: 
\begin{lem}\label{lm:derivatives}
For any degree-$d$ polynomial $f:\R^d \rightarrow \R$, and $x \sim N(0,1)^n$, the following holds with probability at least $1-\delta$: 
\begin{equation}\label{eq:derivatives}
\|\nabla^k p(x)\| \leq O(d^3/\delta) \|\nabla^{k-1} p(x)\|,\;\text{ for all $1 \leq k \leq d$}.
\end{equation}
\end{lem}

Note that the above lemma is tight up to the factor of $O(d^2)$: consider the example $p(x) = x_1^d$.

\paragraph{Independent and concurrent work.} Independently and concurrent to our work, \cite{OSTK21} (following up on \cite{OST20}) also obtained similar results to \cref{th:mainintro}. They first obtained an analogue of \emph{hypervariance reduction} (cf., \cref{lm:hypervarintro}) as studied in \cite{OST20} with better parameters and combined the improved hypervariance reduction lemma with the framework of \cite{OST20} to yield a PRG with $d^{O(1)}$ dependence on the degree $d$. 

Our approach differs in that we critically use our new bounds on the growth of derivatives of polynomials as in \cref{lm:derivatives} (instead of \cref{lm:hypervarintro} which follows from \cref{lm:derivatives}). Working with the derivatives directly allows us to get a substantially simpler analysis of the main PRG construction compared to \cite{OST20, OSTK21}. 

\section{Proof Overview}\label{sec:overview}
We first describe the high-level ideas underlying our main PRG construction - the proof of \cref{th:prgmain}. We then describe the main idea behind the proof of \cref{lm:derivatives} which is critical in being able to handle PTFs of polynomially large degree. The proof of \cref{lm:derivatives} is quite different from the approach taken in \cite{Kane11a, OST20} to prove analogous results in their analysis. 

\subsection{Analysis of the PRG} \label{sec:prgoutline}
We will use the same generator as in \cite{Kane11a}, and the high-level strategy is similar in spirit to that of \cite{Kane11a,OST20}. However, we introduce several additional ingredients that exploit \cref{lm:derivatives} and significantly simplify the analysis. 

As in the works of \cite{Kane11a} and \cite{OST20}, the PRG output will be 
$$Z := \frac{1}{\sqrt{L}}\sum_{i=1}^L Y_{i},$$ where each $Y_i$ is an independent $k$-moment-matching gaussian vector with $k = d^{\Theta(1)}$. For the time being let us work under the idealized assumption that each $Y_i$ is exactly $k$-moment-matching with a standard gaussian: i.e., for any polynomial $h:\R^n \rightarrow \R$ of degree at most $k$, $\E[h(Y_i)] = \E_{z \sim N(0,1)^n}[h(z)]$. We will later relax this condition without too much additional work as is now standard (see \cref{sec:prelims} for details), and ultimately output a discrete approximation to $Z$ with finite support. For now, it is appropriate to imagine that the seedlength required for generating each $Y_i$ will be roughly $O(k \log n)$; the total seedlength will thus be $L \cdot O(k \log n) $. We improve prior works by showing that it suffices to let $L = d^{\Theta(1)}$, rather than $L = 2^{\Theta(d)}$ as in \cite{Kane11a} or $L = d^{\Theta(\log d)}$ as in \cite{OST20}.

For the rest of this section, fix a degree $d$ polynomial $p:\R^n \rightarrow \R$ and let $f:\R^n \rightarrow \{0,1\}$ defined as $f(x) = \sign(p(x))$ be the corresponding PTF we are trying to fool. For simplicity in this introduction, we consider the case where $p$ is multi-linear. The general case is similar but is slightly more nuanced. 

We wish to compare $\E_Z[f(Z)]$ to $\E_z[f(z)]$ where $z \sim N(0,1)^n$. Note that we can rewrite $z \sim N(0,1)^n$ as $z := \frac{1}{\sqrt{L}}\sum_{i=1}^L y_i$ where each $y_i$ is an independent standard gaussian. 
\paragraph{First attempt: A hybrid argument} A natural approach to analyze the PRG is to use a hybrid argument by replacing each $y_i$ with a $k$-moment matching Gaussian vector $Y_i$ as in our PRG output. That is, show the following sequence of inequalities:
\begin{multline}
    \E\left[f\left(\frac{y_1}{\sqrt{L}} + \frac{y_2}{\sqrt{L}} + \cdots + \frac{y_L}{\sqrt{L}}\right)\right] \approx \E\left[f\left(\frac{Y_1}{\sqrt{L}} + \frac{y_2}{\sqrt{L}} + \cdots + \frac{y_L}{\sqrt{L}}\right)\right] \\
    \approx \E\left[f\left(\frac{Y_1}{\sqrt{L}} + \frac{Y_2}{\sqrt{L}} + \cdots + \frac{y_L}{\sqrt{L}}\right)\right] \cdots \approx \E\left[f\left(\frac{Y_1}{\sqrt{L}} + \frac{Y_2}{\sqrt{L}} + \cdots + \frac{Y_L}{\sqrt{L}}\right)\right].
\end{multline}

Let $\lambda = 1/L$ and $y' = \sqrt{\lambda}(y_2 + \cdots + y_L)$. Note that $y' \sim N(0,1-\lambda)^n$. The first inequality in the sequence above, corresponding to a single-step of the hybrid argument is, equivalent to showing: $$\E\left[f(\sqrt{\lambda} y_1 + y')\right] \approx \E \left[f(\sqrt{\lambda} Y_1 + y')\right].$$

In other words, the above inequality is asking to show that $\E[f_{y'/\sqrt{1-\lambda}}(y_1)] \approx \E[f_{y'/\sqrt{1-\lambda}}(Y_1)]$. Intuitively, this is equivalent to showing that $k$-moment matching gaussians fool gaussian restrictions of a PTF with high probability over the restriction. Indeed, such a claim follows from our bounds on the derivatives of polynomials at random evaulation points (\cref{lm:derivatives}).

We say that a polynomial $p$ is \emph{well-behaved} at  a point $x$ if $$\|\nabla^{k+1} p(x)\| \leq (1/\varepsilon) \|\nabla^k p(x)\| \textnormal{ for all } k = 0,1,\ldots,d-1,$$
where $\varepsilon$ is a parameter that will be set to be slightly larger than $\sqrt{\lambda}$. We say $p$ is poorly-behaved at $x$ if the above condition does not hold. 

The starting point of the analysis is that if $p$ is well-behaved at $x$, then $\sign(p(x + \sqrt{\lambda} Y))$ is fooled by a moment-matching $Y$ with \emph{very good} error:

\begin{prop}[Direct Corollary of \cref{lm:moment}] \label{prop:moment-corollary}
Let $q:\R^n \rightarrow \R$ be a degree $d$ multi-linear polynomial and suppose that $q$ is well-behaved at a point $x$. Let $R = \eps^2/\lambda$. Then, for $y \sim N(0,1)^n$ and $Y$ a $dR$-moment matching gaussian,
$$\E_{y \sim N(0,1)^n}[ \sign(q(x + \sqrt{\lambda} y)] - \E_{Y}[\sign(q(x + \sqrt{\lambda} Y)] \leq 2^{-\Omega(R)}.$$ 
\end{prop}

This fact follows from the following argument. Since $q$ is well-behaved at $x$, this in particular implies a non-negligible lower bound on the size of the constant term $c$ of $h(t) := q(x + \sqrt{\lambda} t)$, relative to its other coefficients. In particular, $\sign(h(t))$ is \emph{nearly fixed to a constant} in the sense of \cref{th:rest}. Indeed, writing $h(t) = c + (h(t) - c)$, we see that $\sign(h(t))$ can only differ from $\sign(c)$ if we have a deviation with magnitude at least $|h(t)-c| \geq |c|$. We can use a concentration inequality to bound the probability that either $|h(y) - c| \geq |c|$ or $|h(Y) - c| \geq |c|$. In light of the bounds on $\|\nabla^k q(x)\|$, such a concentration inequality follows from moment bounds obtained from \emph{hypercontractivity}.

 The above lemma shows the first step of the hybrid argument and suggests the following strategy for analyzing the PRG. Define $Z_{-i} = Z - \sqrt{\lambda} Y_i$. We can now aim to show that the polynomial $p$ is well-behaved at $Z_{-i}$ with high probability. This indeed seems plausible as our \cref{lm:derivatives} indeed shows that when $Z$ is standard gaussian, the polynomial $p$ is well-behaved at $Z$ with high probability. 

Immediately, there are two obstacles for this approach: 
\begin{itemize}
    \item First, \cref{lm:derivatives} works only for truly random gaussian and not for our pseudorandom $Z_{-i}$. 
    \item Second, even if we argue that $p$ is likely to be well-behaved at $Z_{-i}$, we cannot apply a union bound over $i$. The error guarantee in \cref{lm:derivatives}, is $ \gg \sqrt{\lambda}$; whereas, we have $L = 1/\lambda$ choices of $i$, so we cannot use such a straightforward union-bound argument to replace each $Y_i$ with a $y_i$.
    
\end{itemize}

The second issue is especially problematic as the error probability in \cref{lm:derivatives} cannot be improved, at least in that variant; the probability that the derivatives don't grow too fast is not small compared to $L = 1/\lambda$.


\paragraph{Beating the union bound.} 
\ignore{For brevity, say that $p$ is \emph{well-behaved} at  a point $x$ if $$\|\nabla^{k+1} p(x)\| \leq (1/\varepsilon) \|\nabla^k p(x)\| \textnormal{ for all } k = 0,1,\ldots,d-1,$$
where $\varepsilon$ is a parameter that will be set to be roughly $\sqrt{\lambda}$. We say $p$ is poorly-behaved at $x$ if the above condition does not hold. If $p$ is well-behaved at $x$, then we know that $\sign(f(x + \sqrt{\lambda} Y))$ is fooled by a moment-matching $Y$ with \emph{very good} error. }

Roughly speaking, the main insight in going beyond the \emph{union bound} obstacle mentioned above is as follows. There are two sources of error in the naive hybrid argument outlined above: (1) The probability of failure coming from $p$ being poorly-behaved at the points $Z_{-i}$. (2) The error coming from applying \cref{prop:moment-corollary} to replace a $Y_i$ with $y_i$ when $p$ is well-behaved at $Z_{-i}$.

Note that we have very good control on the error of type (2) above: we could make it be much smaller than $1/L$ by increasing the amount of independence $k$. We will exploit this critically. We will complement this by showing that even though a naive union bound would be bad for errors of type (1) above, it turns out that we don't have to incur this loss: we (implicitly) show that 
$\Pr\left(\forall i,\; p \text{ is well-behaved at } Z_{-i}\right) \approx 1 - O(\eps d^3)$. 
We do so by checking only that $p$ is well-behaved at the single point $Z$ (in a slightly stronger sense) and then we conclude that $p$ is also highly-likely to be well-behaved at each of the ``nearby" points $Z_{-i}$. Intuitively, this is what allows us to circumvent the union bound in the hybrid argument. However, it would be difficult to actually carry out the analysis as stated this way – we use a sandwiching argument to sidestep the complicated conditionings which would arise in this argument as stated.

We proceed to describe the sandwiching argument. 
We wish to lower-bound the PTF $\sign(p(x))$ by $\sign(p(x)) \cdot g(x)$, where $g(x)$ is some ``mollifier" function taking values in $[0,1]$.
The role of $g(x)$ is roughly to ``test" whether $p$ is well-behaved at $x$; we ideally want $g(x) = 1$ at points $x$ where $p$ is well-behaved and $g(x) = 0$ at points $x$ where $p$ is poorly-behaved. However, we also need $g(x)$ to be smooth, so there will be some intermediate region of points for which $g(x)$ yields a non-informative, non-boolean value. 


We set $g(x)$ to be a smoothed version of the indicator function
$$g(x) \approx  \prod_{k=0}^{d-1} \ind{\|\nabla^{k+1} p(x) \| \leq \frac{1}{\eps} \|\nabla^{k} p(x) \|}, $$
which tests whether the derivatives of $p$ at $x$ have controlled growth in the sense of \cref{lm:derivatives}. 
More specifically, we set
$$ g(x) := \prod_{k=0}^{d-1} \rho \left( \log\left( \frac{1}{ 16 \eps^2}\frac{ \|\nabla^{k} p(x) \|^2 }{\|\nabla^{k+1} p(x) \|^2} \right) \right),
$$
where $\rho(t) : \R \rightarrow [0,1]$ is some smooth univariate function with $\rho(t) = 0$ for $t \leq 0$ and $\rho(t) = 1$ for $t \geq 1$.

Now, for every point $x \in \R^n$ we have
\begin{align*}
    \sign(p(x)) \geq \sign(p(x)) g(x).
\end{align*}
Furthermore, under truly-random gaussian inputs $z \sim N(0,1)^n$ we have
\begin{align*}
    \E_{z} \sign(p(z)) g(z) \geq \E_z \sign(p(z)) - \E_z |g(z) - 1| \geq \E_z \sign(p(z))  - O( \eps d^3),
\end{align*}
where the final inequality here follows from \cref{lm:derivatives}.
Combining these, we get that
\begin{align*}
    \E_Z \sign(p(Z)) \geq \E_z \sign(p(z)) - O( \eps d^3) - |\E_Z \sign(p(Z)) g(Z) - \E_z \sign(p(z)) g(z)|.
\end{align*}
Note that we can similarly obtain an upper-bound for $ \E_Z \sign(p(Z))$ by repeating this argument on the polynomial $-p(x)$.

Thus, it suffices to bound $|\E_Z \sign(p(Z)) g(Z) - \E_z \sign(p(z)) g(z)|$. Having introduced the mollifier, we can now afford to do so by a standard hybrid argument. We represent $z$ as $z := \frac{1}{\sqrt{L}}\sum_{i=1}^L y_i$ and  recall that $Z$ is  of a the form $Z = \frac{1}{\sqrt{L}}\sum_{i=1}^L Y_i$. We can replace each $Y_i$ with $y_i$ and get 
$$ |\E_Z \sign(p(Z)) g(Z) - \E_z \sign(p(z)) g(z)| \leq \gamma L,$$ 
where $\gamma$ is the (quite small) error coming from the following lemma.
\begin{lem}\label{lm:error1}
There exists a constant $c$ such that the following holds for $\lambda \leq \epsilon^2/R d^c$. For any fixed vector $x \in \R^n$, $Y$ a $dR$-moment-matching gaussian vector, and $y \sim N(0,1)^n$, 
$$ |\E_Y \sign(p(x + \sqrt{\lambda} Y)) g(x + \sqrt{\lambda} Y) - \E_y \sign(p(x + \sqrt{\lambda} y)) g(x + \sqrt{\lambda} y) | \leq \gamma = 2^{-\Omega(R)}.$$
\end{lem}
Technically speaking, the above lemma is where our intuition on going around the union bound is quantified, allowing us to use the hybrid argument. 
We briefly outline our proof of this lemma, where for the purpose of illustration we continue with the simplifying assumption that the polynomial $p$ is multilinear. 

The proof is by a case analysis on the behavior of $p$ at the the fixed point $x$. In the multilinear case it suffices to consider the derivatives $\nabla^k p(x)$; in the general case we need to consider something slightly different.
\begin{itemize}
        \item Case 1: $p$ is well-behaved at $x$, i.e., $\|\nabla^{k+1} p(x) \|  \leq (1/\eps) \|\nabla^k p(x) \|$ for all $k$. 
        \begin{itemize}
        \item We can use \cref{lm:moment} in this case to conclude that $\sign(p(x+\sqrt{\lambda} y))$, $\sign(p(x+\sqrt{\lambda} Y))$ are both almost constant with error $2^{-\Omega(R)}$.
        \item So, it remains to show that $Y$ fools $g(x+\sqrt{\lambda} y)$. We approximate $g$ by a low-degree polynomial in $y$ using a Taylor-truncation argument. Our assumption on the controlled growth of derivatives $\|\nabla^{k}p(x)\|$ allows us to bound the truncation error by bounding the higher-moments of the deviations $\|\nabla^k p(x+\sqrt{\lambda} Y) \| - \|\nabla^k p(x) \|$.
        \end{itemize}
        \item Case 2: $p$ is not well-behaved at $x$; let $k_0$ be the largest $k$ such that $\|\nabla^{k_0+1}p(x) \| > (1/\eps)\|\nabla^{k_0} p(x) \| $.
        \begin{itemize}
        \item Intuitively, this says that the polynomial $p$ is well behaved at degree above $k_0$, but not at degree $k_0$. This allows us to show, via an $R$-th moment bound, that both 
        \begin{itemize}
            \item $\|\nabla^{k_0} p(x+\sqrt{\lambda}Y) \| \leq 2 \eps \|\nabla^{k_0+1} p(x)\|$
            \item $\|\nabla^{k_0+1} p(x+\sqrt{\lambda}Y) \| \geq \frac{1}{2} \|\nabla^{k_0+1} p(x)\|$
        \end{itemize}
        are highly likely. Thus, it is highly likely that 
        $$\|\nabla^{k_0} p(x+\sqrt{\lambda}Y)\| \leq 4 \eps \|\nabla^{k_0+1} p(x+\sqrt{\lambda}Y)\|.$$ The latter means $p$ is still sufficiently poorly-behaved at the point $x + \sqrt{\lambda}Y$ that the mollifier classifies it correctly as $g(x+\sqrt{\lambda}Y) = 0.$
        
        \end{itemize}
\end{itemize}


\ignore{
\mnote{It could be good to move the three bullets sketching the proof of the lemma proof to the start of the proof in the main body.}
{ \color{red} Zander: I think if we want additional overview of the proof of the lemma, it is best to put it here -- so I did. Also, this is the only place in the proof overview where we make use of the multilinear simplifying assumption, so I've moved the remark about that down to here.
}}

\begin{ignore}
{\color{red} Zander: Outline for this subsection.

\begin{itemize}
    \item Construction: sum of k-moment-matching gaussians
    \item Analysis: at this point, explain how to at least analyze the idealization $\sqrt{1-\lambda} X + \sqrt{\lambda} Y$ for $X \sim N(0,1)^n$ and $Y$ is $k$-moment-matching.
    \begin{itemize}
        \item Discuss strategy in MZ
        \item Discuss and contrast with strategy in Kane11
        \item Discuss Improvement in OST20
    \end{itemize}
    \item Define hypervariance, and give equivalent operator definition.
    \item Discuss the main use of hypervariance: moment bounds (which are obtained via hypercontractivity). 
    \item With this in mind, give perspective on how the hypervariance reduction lemma can be though of as an analogue of the baby switching lemma.
    \item Elaborate quantitatively on OST20 strategy to analyze $\sqrt{1-\lambda} X + \sqrt{\lambda} Y$ via hypervariance bounds.
    \item At this point, discuss the additional complications involved in analyzing the actual generator (1. we don't have a derandomization of the hypervariance-reduction lemma and 2. we can't afford a union bound.)
    \item Say that, in the spirit of Kane11 and OST20, we handle both of these issues using the a sandwiching argument, but we manage to introduce some simplifications). Give some perspective that the sandwiching argument allows us (in some sense) to test only once whether the hypervariance lemma holds at $Z$, and if it is well-behaved at $Z$ we infer that it must also be nearly as well-behaved at each of the nearby points $Z_{-i}$.
    \item For simplicity we describe the sandwiching argument for the multilinear case:
    \begin{itemize}
        \item We seek to lower bound the generator's value on the polynomial step function $\ind{p(x) \geq 0 }$ -- explain why this is sufficient. (This is a technicality but still too important to ignore even here. The sandwiching argument does not work with the typical definition of $\sign(p(x))$.) We should just PTF to be $\ind{p(x) \geq 0}$ and call that the sign function. 
        \item We define the mollifier function g(x), which softly checks whether the polynomial $p(x)$ is well-behaved at $x$, with respect to our more specific hypervariance-reduction lemma which compares the 2-norm of each consecutive pair of derivatives.  
        \item Highlight some specific properties of the mollifier: it takes values in $[0,1]$, it takes a value of $0$ if the polynomial is sufficiently bad, and it is smooth. Say that we will set $\eps \approx \sqrt{\lambda}$.
        \item Describe the lower-sandwiching argument, from which we conclude that the error is at most $O(\eps d^3)$ + $L \cdot \gamma$, where $\gamma$ is the error with which a single $k$-moment-matching $Y$ fools $\sign(p(z+\sqrt{\lambda} Y)) g(z+\sqrt{\lambda} Y)$ for arbitrary fixed $z$. Thus, we can afford to pay an error of $\gamma$ $L$ times, but we could not pay an error of size $\gg \sqrt{\lambda} = 1/\sqrt{L}$ $L$ times.
        \end{itemize}
    \item Describe how to analyze the error for $\sign(p(z+\sqrt{\lambda} Y)) g(z+\sqrt{\lambda} Y)$, when $Y$ is $k$-moment-matching.
    \begin{itemize}
        \item fix a truncation threshold $R = \Theta(\log(1/\gamma)$.
        \item Case 1: $\|\nabla^k p(z) \| \geq \eps \|\nabla^{k+1} p(z) \| $ for all $k$.
        \item In this case, the non-constant part of $p(z+\sqrt{\lambda} z)$ rarely exceeds it's constant term in magnitude. So  $\sign(p(z+\sqrt{\lambda} Y))$ is constant with high probability, and since this is true by an $R-th$ moment bound, it is also true if $Y$ is $dR$-moment-matching.
        \item So, it remains to show that $Y$ fools $g(z+\sqrt{\lambda} Y)$. We approximate $g$ by a low-degree polynomial in $Y$ using a Taylor-truncation argument. Our assumption on the well-behavedness of $\|\nabla^k p(z) \|$ allows us to bound the Taylor-truncation error by bounding the higher moments of the deviations $\|\nabla^k p(z+\sqrt{\lambda} Y) \| - \|\nabla^k p(z) \|$.
        \item Case 2: $k_0$ is the largest $k$ such that $\|\nabla^{k_0} p(z) \| < \eps \|\nabla^{k_0+1} p(z) \|$.
        \item Intuitively, this says that the polynomial $p$ is well behaved at degree above $k_0$, but not at degree $k_0$. This allows us to show, via an $R$-th moment bound, that both 
        \begin{itemize}
            \item $\|\nabla^{k_0} p(z+\sqrt{\lambda}Y) \| \leq 2 \eps \|\nabla^{k_0+1} p(z)\|$ 
            \item $\|\nabla^{k_0+1} p(z+\sqrt{\lambda}Y) \| \geq \frac{1}{2} \|\nabla^{k_0+1} p(z)\|$
        \end{itemize}
        are highly likely. Thus, it is highly likely that $\|\nabla^{k_0} p(z+\sqrt{\lambda}Y)\| \leq 4 \eps \|\nabla^{k_0+1} p(z+\sqrt{\lambda}Y)\|$, and so $g(z+\sqrt{\lambda}Y) = 0$. Since this is proved by an $R$-th moment bound, it is true even if $Y$ is merely $k$-moment-matching. 
    \end{itemize}
\end{itemize}
}
\end{ignore}

\subsection{Slow-growth of derivatives and simplification under restrictions}
The proof of \cref{lm:derivatives} is iterative and is relatively simple given Kane's \emph{relative anti-concentration inequality} for degree $d$ polynomials \cite{Kane13} developed in the context of studying the \emph{Gotsman-Linial} conjecture for PTFs. 

\cite{Kane13} shows that for any degree $d$ polynomial, and $x,y \sim N(0,1)^n$ with probability at least $1-\delta$, we have $|\langle y, \nabla p(x) \rangle| \leq (d^2/\delta) |p(x)|$. As $y$ in the above statement is independent of $x$, for any $x$, $\langle y, \nabla p(x) \rangle$ is distributed as $N(0, \|\nabla p(x)\|^2)$. This says that the inequality is essentially equivalent to saying that with probability at least $1-\delta$ over $x$, we have $\|\nabla p(x)\|^2 \leq O(d^2/\delta) |p(x)|$. The latter can be seen as the inequality corresponding to $k=1$ in the statement of \cref{lm:derivatives}. The full proof of the lemma is via iteratively applying the above argument using a  vector-valued generalization of Kane's inequality. 

Next, it is not too hard to prove \cref{th:rest} given \cref{lm:derivatives}. For illustration, suppose that we have a degree $d$ multi-linear polynomial $p$, and write $f(t) := p(\sqrt{1-\lambda} t)$. Then, by elementary algebra\footnote{If $p$ is multi-linear, then the Hermite expansion (see \cref{sec:prelims}) is just $p(x) = \sum_{\alpha \in \{0,1\}^n } \hat{p}(\alpha) h_{\alpha}(x) = \sum_{I\subseteq [n]} \hat{p}(I) \prod_{i \in I } x_i$. We can prove the identity for each monomial and use additivity.}, we have the identity
\begin{equation}\label{eq:id1}
p_x(y) = p\left(\sqrt{1-\lambda} x + \sqrt{\lambda} y\right) = \sum_{\alpha} \partial^\alpha f(x) \left(\frac{\lambda}{1-\lambda}\right)^{|\alpha|/2} y^{\alpha}.
\end{equation}

Now, by \cref{lm:derivatives}, with probability $1-\delta$ over $x$, we have $\|\nabla^k f(x)\| \leq O(d^3/\delta) \|\nabla^{k-1} f(x)\|$, for all $k$. Thus, if we take $\lambda \ll \delta^2 / (R^2 d^6)$, the factor of $\lambda$ will kill the growing derivatives leading to a bound on the higher-order moments of $p_x(y) - p_x(0)$ via hypercontractivity. These moment bounds in turn imply that $|p_x(y) - p_x(0)| < |p_x(0)|$ with high probability over $y$, and hence that $\sign(p_x(y)) = \sign(p_x(0))$ with high probability over $y$. 

Notice that \cref{eq:id1} is essentially a Taylor expansion of $p$ at $\sqrt{1-\lambda}x$: it expresses the function $p_x(y)$ as a polynomial in $y$ in the standard basis, whose coefficients are determined by the derivatives of $p$ at $\sqrt{1-\lambda}x$. 
We want to do something similar in the general case, but in the Hermite basis; for non-multi-linear polynomials these two bases no longer coincide. 
So, in the general case, we rely on the following identity, which we regard as an analogue of the Taylor expansion for the Hermite basis. 

\begin{lem}[See \cref{sec:prelims}]
Let $f(y) = \sum_{\alpha} \hat{f}(\alpha) h_{\alpha}(y).$ Then
$$ f\left( \sqrt{1-\lambda} x + \sqrt{\lambda} y  \right) = \sum_{\alpha} \frac{\partial^{\alpha} g(x)}{\sqrt{\alpha !}} \left(\frac{\lambda}{1-\lambda} \right)^{|\alpha|/2} h_{\alpha}(y) ,$$
where $g(x) := U_{\sqrt{1-\lambda}} f(x) = \sum_{\alpha} \hat{f}(\alpha) (1-\lambda)^{|\alpha|/2} h_{\alpha}(x).$
\end{lem}

Hermite polynomials are such a ubiquitous tool used in such a wide range of fields that it seems unlikely that such an identity is new. However, we are not aware of any previous appearance of such an identity in the literature (at least in the body of work on PTFs) and we provide a proof. 

\paragraph{Hypervariance reduction.} We next remark on the relation between \emph{slow-growth of derivatives} (as in \cref{lm:derivatives}) and \emph{hypervariance reduction} as studied and introduced in \cite{OST20}. The latter plays a similar role in their paper as the former does in this work. However, \cref{lm:derivatives} importantly has only polynomial dependence on the degree $d$ and is also much more conducive to our analysis of the PRG.

Recall the Hermite expansion (see \cref{sec:prelims}) of polynomials: A degree $d$ polynomial $p:\R^n \rightarrow \R$ can be uniquely expressed as 
$$ p(y) := \sum_{|\alpha| \leq d} \hat{p}(\alpha) h_{\alpha}(y), $$
where $\alpha \in \N^n$ denotes a multi-index and $h_\alpha(y)$ is the $\alpha$'th Hermite polynomial. The \emph{hypervariance} and  \emph{normalized hypervariance} of a polynomial introduced in \cite{OST20} are defined as follows:

\begin{definition}
For a polynomial $p:\R^n \rightarrow \R$ of the form $ p(y) := \sum_{\alpha} \hat{p}(\alpha) h_{\alpha}(y)$, define its \emph{hypervariance}, $\hvar_R(\;)$, and \emph{normalized hypervariance}, $H_R(\;)$, as 
$$\hvar_R(p) := \sum_{\alpha \neq 0} \hat{p}(\alpha)^2 R^{2|\alpha|},\;\; H_R(p) := \frac{\hvar_R(p)}{\hat{p}(0)^2}.$$
\ignore{define its \emph{hypervariance} as
$$ \textnormal{HyperVar}_R(p) := \sum_{\alpha \neq 0} \hat{p}(\alpha)^2 R^{2|\alpha|}. $$
We define the \emph{normalized hypervariance} as
$$ H_{R}(p) := \frac{\textnormal{HyperVar}_R(p) }{\hat{p}(0)^2} .$$}
\end{definition}

Intuitively, if the normalized hypervariance $H_R(p)$ of a polynomial is small for a large $R$, then it means that the \emph{weights} of the higher-order Hermite coefficients of $p$ have a geometric decay. 

\cite{OST20} showed that for any polynomial $p$, for a suitable $\lambda > 0$, a gaussian restriction of $p$ will have small normalized hypervariance with high probability. Specifically, they showed that if $\lambda = d^{-O(\log d)}$, then $H_R(p_{x,\lambda})$ is bounded with high probability over $x \sim N(0,1)^n$. They also asked whether this property holds when $\lambda = d^{-O(1)}$ instead of being quasi-polynomially small in $d$. \cref{lm:derivatives} implies this conjecture without too much difficulty: 

\begin{lem}\label{lm:hypervarintro}
For any degree $d$ polynomial $p$ and $\lambda, \delta > 0$, the following holds. Except with probability $\delta$ over $x \sim N(0,1)^n$, the normalized hypervariance $H_R(p_{x,\lambda}) = O(\lambda d^6 R^2/\delta^2)$. 
\end{lem}

The proof of the analogue of \cref{lm:hypervarintro} for quasi-polynomially small $\lambda$ (i.e.\ $\lambda = d^{-O(\log d)}$) in \cite{OST20} was by an iterative process: Intuitively, if one sets $\lambda_0 = d^{-O(1)}$, and $\lambda = \lambda_0^{\log d}$, then the random restriction $p_{\lambda,x}$ is equivalent to $(\log d)$ independent random restrictions with restriction parameter $\lambda_0$. The authors in \cite{OST20} show that each such $\lambda_0$-restriction (essentially) decreases the degree by a factor of $2$. 
We instead take a different approach by drawing a connection between norms of derivatives and to \emph{relative anti-concentration} as developed in the context of studying the \emph{Gotsman-Linial} conjecture for PTFs.

\section{Preliminaries} \label{sec:prelims}
\textbf{The pseduorandom generator construction: idealization vs.\ discretization.} Following \cite{Kane11a} and \cite{OST20}, we analyze the idealized pseudorandom distribution
$$Z = \frac{1}{\sqrt{L}}\sum_{i=1}^L  Y_i,$$
where each $Y_i \in \R^n$ is a $k$-moment-matching gaussian (that is, $\E[p(Y_i)] = \E_{x \sim N(0,1)^n}[p(x)]$ for all polynomials $p : \R^n \rightarrow \R$ of degree at most $k$).

Suppose that, for any such $Z$ with parameters $(L, k)$, it is the case that $Z$ fools degree-$d$ PTFs with error $\eps = \eps(L,k,d)$.
Then, it is shown in \cite{Kane11a} how to obtain a small-seedlength PRG (in the sense of \cref{def:prg}) by providing a specific instantiation and discretization of this construction.

\begin{thm}[\cite{Kane11a}, implicit in Section 6]\label{th:discretize}
Suppose a $Z$ as above with parameters $(L,k)$ fools degree $d$-PTFs with error $\varepsilon = \eps(L,k,d)$. Then, there is an explicit, efficiently computable PRG with seedlength $O(d k L \log(ndL/\eps)$ that $(2\eps)$-fools degree $d$ PTFs. 
\ignore{that fools degree-$d$ PTFs on $n$-variables with seedlength
$$ r = O(dkL \log(ndL/\eps)) $$
and error
$$ \eps' \leq 2\eps .$$}
\end{thm}

\textbf{Hermite polynomials.} To argue about polynomials which are not necessarily multilinear, we need some simple facts concerning Hermite polynomials. For our purposes, Hermite polynomials are simply a convenient choice of polynomial basis which have nice properties (in particular being \emph{orthonormal}) with respect to gaussian inputs. For a more detailed background on Hermite polynomials and their use for analyzing functions over gaussian space, see \cite[Ch.\ 11]{ODonnel14}.

One concrete way to define the Hermite polynomials is the following:
\begin{itemize}
    \item For the univariate polynomials, the degree-$m$ ``Probabilist's" Hermite polynomial is the $m$-th coefficient of the generating function
    $$ e^{st - \frac{1}{2}s^2} = \sum_{m\geq 0} H_m(t) s^m. $$
    \item We define the degree-$m$ univariate Hermite polynomial by the normalization
    $$ h_m(t) := \frac{1}{\sqrt{m!} }  H_m(t).$$
    \item For a multi-index $\alpha \in \N^n$, we define the multivariate Hermite polynomial $h_{\alpha} : \R^n \rightarrow \R$ via the product
    $$ h_\alpha(x) := \prod_{i=1}^n h_{a_i}(x_i) .$$
\end{itemize}

We record some basic properties of this particular choice of polynomial basis. The final two properties say that the Hermite basis is orthonormal with respect to correlation under the standard gaussian distribution -- this is the reason for our choice of normalization.

\begin{itemize}
    \item The set $\set{h_\alpha(x)}{|\alpha| \leq d}$ is a basis for real polynomials in $n$ variables of degree $\leq d$.
    \item $h_0$ is the constant polynomial $h_0 \equiv 1$.
    \item For multi-indicies $\alpha \in \{0,1\}^n$, $h_\alpha(x)$ is simply the monomial $\prod_{i : a_i = 1} x_i$.
    \item For $x \sim N(0,1)^n$, and distinct multi-indices $\alpha \neq \beta$, 
    $ \E_x h_{\alpha}(x) h_{\beta}(x) = 0. $
    \item For $x \sim N(0,1)^n$, and any multi-index $\alpha$, $\E_x h_{\alpha}(x)^2 = 1$.
\end{itemize}
\ignore{
$$ e^{st - \frac{1}{2}s^2} = \sum_{m \geq 0 } \sqrt{m!} h_m(t) s^m $$
for univariate Hermite polynomials would be a consequence.}

\textbf{Guassian noise operator.} We recall the definition of the noise operator $U_{\rho}$, which here we regard as an operator on real polynomials in $n$ variables (see \cite[Ch.\ 11]{ODonnel14} for background and a more general viewpoint).
For a polynomial $f : \R^{n} \rightarrow \R$ and a parameter $\rho \in [0,1]$, the action of $U_{\rho}$ on $f$ is specified by
$$ (U_\rho f)(x) := \E_{Z \sim N(0,1)^n} f\left(\rho x + \sqrt{1-\rho^2} Z\right) .$$
An important feature of the Hermite basis is that the noise operator acts on it \emph{diagonally} (see \cite[Ch.\ 11]{ODonnel14}):
$$ U_\rho h_{\alpha}(x) = \rho^{|\alpha|} h_{\alpha}(x) .$$
Thus, if $f$ is a degree-$d$ polynomial given in the Hermite basis as
$$ f(x) = \sum_{|\alpha| \leq d} \hat{f}(\alpha) h_{\alpha}(x) ,$$
then we can express the result of the noise operator applied to $f$ explicitly as
$$ U_\rho f(x) = \sum_{|\alpha| \leq d}  \hat{f}(\alpha) \rho^{|\alpha|} h_{\alpha}(x) .$$

\textbf{Higher moments and hypercontractivity.}
Fix a polynomial $f(x) := \sum_{|\alpha| \leq d} \hat{f}(\alpha) h_{\alpha}(x)$.
For an even natural number $q \geq 2$, we write the gaussian $q$-norm of $f$ as
$$ \|f\|_{q} := \left( \E_{x \sim N(0,1)^n} f(x)^q \right)^{1/q} .$$
We wish to be able to bound this quantity in terms of the magnitudes of the Hermite coefficients of $f$, $\hat{f}(\alpha)$. 
For this purpose, we extend the definition of $U_\rho$ also to $\rho > 1$ by its action on the Hermite basis:
$ U_\rho h_{\alpha}(x) = \rho^{|\alpha|} h_{\alpha}(x) .$ 
With this notation, we can express the well-known $(q,2)$-hypercontractive inequality \cite[Ch.\ 9,11]{ODonnel14} as
$$ \|f\|_{q} \leq \| U_{\sqrt{q-1}} f \|_2 ,$$
which is quite convenient for us, as we can use orthonormality of the Hermite basis to explicitly compute
$$ \| U_{\sqrt{q-1}} f \|_2^2 = \sum_{|\alpha| \leq d} (q-1)^{|\alpha|} \hat{f}(\alpha)^2 \leq \sum_{|\alpha| \leq d} q^{|\alpha|} \hat{f}(\alpha)^2.$$

To get a feel for the utility of this bound, let's see how it can be used to prove the following concentration bound:

\begin{lem} \label{lm:moment}
Let $f:\R^n \rightarrow \R$ be a degree $d$ polynomial with normalized hypervariance $H_{\sqrt{q}}(f) \leq \frac{1}{4}$, where $q$ is an even natural number. Then, 
$$\P_{y \sim N(0,1)^n} \paren{\sign(f(y)) \neq \sign(\hat{f}(0))} \leq 2^{-q}.$$
Further, the same holds more generally for $y \sim Y$, as long as the distribution $Y$ is $dq$-moment-matching. 
\end{lem}
\begin{proof}
Suppose that $f(y)$ is normalized so that
$$ \E_{y \sim N(0,1)^n} f(y) = \hat{f}(0) = \pm 1 .$$
We have the $q$-th moment bound
$$ \|f(x)- \hat{f}(0)\|_{q} \leq \|U_{\sqrt{q}} \left( f(y) - \hat{f}(0) \right) \|_2   \leq \tfrac{1}{2}. $$
From the generic concentration inequality
$$ \P \paren{|X| \geq t \|X\|_q} \leq t^{-q}$$
we obtain
$$ \prob{\sign(f(y)) \neq \sign(\hat{f}(0))} \leq 2^{-q}. $$ 
Thus, we find that the PTF $\sign(f)$ almost always yields the value $\sign(\hat{f}(0))$ under random gaussian inputs.
Crucially for us, this argument is also \textit{easy to derandomize}: since the argument merely relies on a bound on the $q$-th moment $\E_{y \sim N(0,1)^n} (f(y) - \hat{f}(0))^q$, and for $Y$ which is $k$-moment-matching for $k \geq dq$ we have
$$ \E_{Y} (f(Y) - \hat{f}(0))^q = \E_{y \sim N(0,1)^n} ( f(y)- \hat{f}(0))^q,$$
we conclude also that $\sign(f(Y))$ is typically equal to $\sign(\hat{f}(0))$.
\end{proof}
We remark that this lemma further implies that $Y$ fools $\sign(f)$ when $H_{\sqrt{q}}(f)$ is small:
\begin{equation*}
\E_Y \sign(f(Y)) = \E_{y \sim N(0,1)^n} \sign(f(y)) \pm O(2^{-q}). 
\end{equation*}

\textbf{Gaussian restrictions and derivatives on the Hermite basis.} Besides the effect of the noise operator, it will also be important to understand the effect of two further operations on polynomials: 
\begin{itemize}
    \item The derivative map, $ f(y) \mapsto  \partial^{\alpha} f(y) .$
    \item The gaussian restriction at $x$, $f(y) \mapsto f\left(\sqrt{1-\lambda} x + \sqrt{\lambda}y\right)$.
\end{itemize}
In particular, we are concerned with how these operations affect the Hermite coefficients of a polynomial; ultimately, our goal will be to develop a ``Hermite-basis analogue" of the Taylor expansion which can be applied to expand $f\left(\sqrt{1-\lambda} x + \sqrt{\lambda}y\right)$ as a function of $y$.
We start by computing the effect of these two operations on univariate Hermite polynomials, and then on the full multivariate Hermite basis, and finally on a general polynomial $f(x)$ expressed in the Hermite basis.

\begin{prop}
For univariate Hermite polynomials, we have the identities
\begin{itemize}
\item $ \frac{\partial^k}{\partial t^k} h_{m}(t) = \sqrt{\frac{m!}{(m-k)!}}h_{m-k}(t)$,
\item $ h_{m}\left( \sqrt{1-\lambda} x + \sqrt{\lambda} y \right) = \sum_{k=0}^{m} \sqrt{\binom{m}{k}}  (1-\lambda)^{(m-k)/2}  \lambda^{k/2} h_{m-k}(x) h_{k}(y). $
\end{itemize}
\end{prop}
\begin{proof}
The first of these identities is standard (see e.g. \cite[Ex.\ 11.10]{ODonnel14}); we provide a proof of the second.

The second identity can be proved by considering the generating function
$$ e^{st - \frac{1}{2}s^2} = \sum_{m} \sqrt{m!} h_{m}(t) s^m, $$
and comparing the coefficient of $s^m$ on both sides of
\begin{equation*}
 e^{s(\sqrt{1-\lambda} x + \sqrt{\lambda} y) - \frac{1}{2}s^2} = e^{(s\sqrt{1-\lambda})x - \frac{1}{2} (s\sqrt{1-\lambda})^2}  \cdot e^{(s\sqrt{\lambda})y - \frac{1}{2} (s\sqrt{\lambda})^2} \qedhere
\end{equation*}
\end{proof}
The corresponding identities for multivariate Hermite polynomials follow easily from above.
\begin{prop} We have
\begin{itemize} 
\item $ \partial^{\alpha} h_{\beta}(y) = \sqrt{\frac{\alpha !}{\gamma !}}h_{\gamma}(y)$, where $\gamma = \beta - \alpha$,
\item $ h_{\beta}\left( \sqrt{1-\lambda} x + \sqrt{\lambda} y \right) = 
(1-\lambda)^{|\beta|/2} \sum_{\alpha \leq \beta} \frac{\partial^{\alpha} h_\beta(x)}{\sqrt{\alpha !}} \left(
\frac{\lambda}{1-\lambda}\right)^{|\alpha|/2} h_{\alpha}(y)$,
\item $ \partial^{\alpha} h_{\beta}\left( \sqrt{1-\lambda} x + \sqrt{\lambda} y \right) = 
(1-\lambda)^{|\beta - \alpha|/2} \sum_{\gamma \leq \beta - \alpha} \frac{\partial^{\alpha + \gamma} h_{\beta}(x)}{\sqrt{\gamma!}} \left(\frac{\lambda}{1-\lambda}\right)^{|\gamma|/2} h_{\gamma}(y).
$
\end{itemize}
\end{prop}

We conclude with a Taylor-like expansion in the Hermite basis that we use repeatedly.

\begin{lem} \label{lm:hermite1}
Let $f(y) = \sum_{\alpha} \hat{f}(\alpha) h_{\alpha}(y).$ Then
$$ f\left( \sqrt{1-\lambda} x + \sqrt{\lambda} y  \right) = \sum_{\alpha} \frac{\partial^{\alpha} g(x)}{\sqrt{\alpha !}} \left(\frac{\lambda}{1-\lambda} \right)^{|\alpha|/2} h_{\alpha}(y) ,$$
where $g(x) := U_{\sqrt{1-\lambda}} f(x) = \sum_{\alpha} \hat{f}(\alpha) (1-\lambda)^{|\alpha|/2} h_{\alpha}(x).$
\end{lem}

\begin{proof}

We express
\begin{align*}
    f\left( \sqrt{1-\lambda} x + \sqrt{\lambda} y \right) &= \sum_{\alpha} \hat{f}(\alpha) h_{\alpha}\left( \sqrt{1-\lambda} x + \sqrt{\lambda} y \right) \\
    &= \sum_{\alpha}  \frac{h_{\alpha}(y)}{\sqrt{\alpha !}} \left(\frac{\lambda}{1-\lambda} \right)^{|\alpha|/2} \sum_{\beta \geq \alpha} \hat{f}(\beta) (1-\lambda)^{|\beta|/2}\partial^{\alpha} h_{\beta}(x) \\
    &= \sum_{\alpha} \frac{h_{\alpha}(y)}{\sqrt{\alpha !}} \left(\frac{\lambda}{1-\lambda} \right)^{|\alpha|/2} \partial^{\alpha} g(x)   . \qedhere
\end{align*}
\end{proof}

Lastly, we will also need an extension of this theorem which expresses $\partial^{\alpha} f$, at the point $$ \sqrt{1-\lambda} x + \sqrt{\lambda} y,$$ as a polynomial in $y$ in the Hermite basis.

\begin{thm} \label{lm:hermite2}
Let $f(y) = \sum_{\alpha} \hat{f}(\alpha) h_{\alpha}(y).$ Then
$$ \partial^{\alpha} f\left( \sqrt{1-\lambda} x + \sqrt{\lambda} y  \right) = (1 - \lambda)^{-|\alpha|/2}\sum_{\beta \geq \alpha} \partial^{\beta} g(x) \sqrt{\frac{\alpha !}{\beta !}} \left(\frac{\lambda}{1-\lambda} \right)^{|\beta-\alpha|/2} h_{\beta - \alpha}(y) ,$$
where $g(x) := U_{\sqrt{1-\lambda}} f(x).$
\end{thm}
\begin{proof}
We express
\begin{align*}
    \partial^{\alpha} f\left( \sqrt{1-\lambda} x + \sqrt{\lambda} y \right) &= \sum_{\beta} \hat{f}(\beta)  \partial^{\alpha} h_{\beta}\left( \sqrt{1-\lambda} x + \sqrt{\lambda} y \right) \\
    &= \sum_{\gamma} \frac{h_{\gamma}(y)}{\sqrt{\gamma !}} \left(\frac{\lambda}{1-\lambda} \right)^{|\gamma|/2} \sum_{\beta \geq \gamma + \alpha} (1-\lambda)^{|\beta-\alpha|/2} \partial^{\alpha+\gamma} h_{\beta}(x) \\
    &=   (1 - \lambda)^{-|\alpha|/2} \sum_{\gamma} \frac{h_{\gamma}(y)}{\sqrt{\gamma !}} \left(\frac{\lambda}{1-\lambda} \right)^{|\gamma|/2} \partial^{\alpha + \gamma} g(x). \qedhere
\end{align*}
\end{proof}

\section{Gaussian restrictions of polynomials}
Here we prove the structural properties of gaussian restrictions of polynomials: \cref{th:rest}, \cref{lm:derivatives}, \cref{lm:hypervarintro}. Note that \cref{th:rest} follows immediately from \cref{lm:hypervarintro} and \cref{lm:moment}. We next prove \cref{lm:hypervarintro} from \cref{lm:derivatives}.

\begin{proof}[Proof of \cref{lm:hypervarintro} from \cref{lm:derivatives}]

Define $f(x) := U_{\sqrt{1-\lambda}} p(x)$.
Then, by \cref{lm:hermite1},
\begin{align*}
    p_{x}(y) &= f(x) + \sum_{\alpha \neq 0} \frac{\partial^{\alpha} f(x)}{\sqrt{\alpha !}} \left(\frac{\lambda}{1-\lambda} \right)^{|\alpha|/2} h_{\alpha}(y). 
\end{align*}
Thus,
\begin{align*}
    \hvar_R(p_x) = \sum_{\alpha \neq 0} \left(\frac{\partial^{\alpha} f(x)}{\sqrt{\alpha !}}\right)^2 \left(\frac{\lambda}{1-\lambda} \right)^{|\alpha|} R^{2|\alpha|} &\leq \sum_{\alpha \neq 0} \left(\partial^{\alpha} f(x)\right)^2 \left(\frac{\lambda}{1-\lambda} \right)^{|\alpha|} R^{2|\alpha|} \\&= \sum_{k=1}^d R^{2k} \left(\frac{\lambda}{1-\lambda}\right)^k \|\nabla^k f(x)\|^2,
\end{align*}
where the first inequality follows as $\sqrt{\alpha!} \geq 1$. 

We now conclude by applying \cref{lm:derivatives} to $f$. We have
$$ H_R(p_x) = \frac{\sum_{k=1}^{d} R^{2k} \left(\frac{\lambda}{1-\lambda}\right)^k \| \nabla^k f(x) \|^2  }{f(x)^2} .$$
Except with probability $\delta$ over $x \sim N(0,1)^n$, we can bound this by
\begin{equation*} \sum_{k=1}^d R^{2k} \left(\frac{\lambda}{1-\lambda}\right)^k \left( \frac{C d^3}{\delta}\right)^{2k} \leq  O\left(\frac{\lambda d^6 R^2}{\delta^2} \right). \qedhere
\end{equation*}
\end{proof}

\subsection{Proof of \texorpdfstring{\cref{lm:derivatives}}{Lemma 1.8}}
Our main tool will be Kane's relative-anticoncentration lemma for gaussian polynomials

\begin{lem}[\cite{Kane13}]
For a degree $d$ polynomial $p$, and independent standard gaussian vectors\\ $x,y \in \R^n$, 
$$\prob{ |p(x)| \leq \eps | \ip{y}{\nabla p(x)} |} \leq O(\eps d^2).$$
\end{lem}
In fact, we will actually work with the following corollary which is essentially the first of the $d$ inequalities in \cref{lm:derivatives}. 
\begin{cor} \label{cor:kane_cor}
For a degree $d$ polynomial $p$, and independent standard gaussian vector $x \in \R^n$,
$$\prob{ |p(x)| \leq \eps \| \nabla p(x) \| } \leq O(\eps d^2) .$$
\end{cor}
\begin{proof}
We note that for any fixed $x$, $\ip{y}{\nabla p(x)}$ is identical in distribution to $Z  \| \nabla p(x) \|$, 
where $Z \sim N(0,1)$ is a standard gaussian. 
So, we express
\begin{align*}
    \prob{ |p(x)| \leq \eps | \ip{y}{\nabla p(x)} |} &= 
    \prob{ |p(x)| \leq \eps |Z|  \| \nabla p(x) \| } \\
    &\geq \prob{|p(x)| \leq \eps \| \nabla p(x) \| } \cdot \prob{|Z| \geq 1}.
\end{align*}
Since $\prob{|Z| \geq 1} \geq \Omega(1)$, we conclude that
\begin{equation*}
 \prob{|p(x)| \leq \eps \| \nabla p(x) \| } \leq O(\eps d^2) .\qedhere
\end{equation*}
\end{proof}

The heart of the proof of \cref{lm:derivatives} is a vector-valued variant of the above corrollary:

\begin{lem} \label{lm:main}
Let $\vv{f}(x) := \left(f_1(x), f_2(x), \ldots, f_m(x)\right)$ be a collection of $m$ degree-at-most $d$ polynomials $f_j(x)$. If $x \in \R^n$ is a standard gaussian vector, then
$$ \prob{ \|\vv{f}(x)\|^2 \leq \eps^2 \sum_{j=1}^m \|\nabla f_j(x) \|^2 } \leq O(\eps d^2) .$$
\end{lem}

\begin{proof}[Proof of \cref{lm:derivatives}]
We simply apply the above lemma $d$ times and take a union bound. For $1 \leq k \leq d$, let $\vv{f}_k(x) := ((\partial^\alpha f(x): |\alpha| = k))$. Note that $\|\vv{f}_k(x)\|^2 = \|\nabla^k f(x)\|^2$. Further, note that
$$\sum_{\alpha: |\alpha| = k} \|\nabla (\partial^\alpha f(x))\|^2 \geq  \|\nabla^{k+1} f(x)\|^2,$$
where the inequality follows as each $(k+1)$'th order derivative would be counted at least once in the expression on the left hand side. Therefore, by the above lemma, for $x \sim N(0,1)^n$, we have
$$\prob{ \|\nabla^k f(x)\|^2 \leq \eps^2  \|\nabla^{k+1} f(x)\|^2 } \leq O(\eps d^2) $$

Setting $\eps = \delta/d^3$, and taking a union bound over all $k$, we get that for a constant $C > 0$, 
$$\prob{\forall k, \|\nabla^k f(x)\|^2 > C (\delta^2/d^6) \|\nabla^{k+1} f(x)\|^2} \geq 1 - \delta.$$

This proves \cref{lm:derivatives}.
\end{proof}

\begin{proof}[Proof of \cref{lm:main}.]
Consider the auxiliary polynomial 
$$h(x,y) := \sum_{j=1}^m f_j(x) y_j.$$
As a function of both $x$ and $y$, we have
$$ \nabla h(x,y) = \vv{f}(x) \circ M_x y ,$$
where $M_x$ is the matrix with columns $\nabla f_j(x)$ (that is, $M_x$ has $(i,j)$-th entry $\frac{\partial}{\partial x_i} f_j(x)$).
So, applying \cref{cor:kane_cor} to this auxiliary polynomial gives the probability bound
\begin{align*}
 q &:= \prob{h(x,y)^2 \leq \eps^2 \| \nabla g(x,y) \|^2}  \\
 &= \prob{ \ip{y}{\vv{f}(x)}^2 \leq \eps^2 \left( \|\vv{f}(x)\|^2 + \|M_x y\|^2 \right)} \\
 &\leq O(\eps d^2).
\end{align*}
Now, for some constant $C \geq 2$ to be specified later, let $E$ denote the event that
$$ (C^2 - 1) \|\vv{f}(x)\|^2 \leq \frac{\eps^2}{2} \|M_x\|_{F}^2 ,$$
where $\|M_x\|_{F}$ is the Frobenius norm of $M_x$. We note that we can lower-bound the probability $q$ by
$$ q \geq \prob{E} \cdot \prob{\left|\ip{y}{\vv{f}(x)}\right| \leq C \|\vv{f}(x)\| \textnormal{ and } \|M_x y\|^2 \geq \frac{1}{2} \|M_x\|_F^2 \; | E }. $$
We claim that for large enough choice of constant $C$, this conditional probability can be lower-bounded by $\Omega(1)$. Indeed, we can argue for any fixed $x$: 
\begin{itemize}
    \item $\prob{\left|\ip{y}{\vv{f}(x)}\right| \geq C \|\vv{f}(x)\|} \leq \frac{1}{C^2}$.
    \item $\prob{\|M_x y\|^2 \geq \frac{1}{2} \|M_x\|_F^2} \geq \Omega(1)$.
\end{itemize}
The first item is just a Chebyshev inequality; the second item can be derived e.g.\ from the basic anticoncentration bound one obtains for degree-2 polynomials from the Paley-Zygmund bound together with hypercontractivity (since, for any fixed matrix $M$, the quadratic form $g(y) := \|M y\|^2$ has second-moment $\E g(y)^2 \geq (\E g(y))^2 = \|M\|_F^2$).

Thus, by choosing $C$ large enough, we can lower-bound this conditional probability by 
$$ \Omega(1)  - \frac{1}{C^2} \geq \Omega(1) .$$
We conclude that $\prob{E} \leq O(q) = O(\eps d^2)$.
This gives the desired conclusion
\begin{equation*}
 \prob{\|\vv{f}(x)\| \leq \Omega(\eps) \|M_x \|_F} \leq O(\eps d^2) . \qedhere
\end{equation*}
\end{proof}


\section{Pseudorandom Generator for PTFs} \label{sec:prg}
The following theorem gives quantitative bounds on the error of our main generator: 
\begin{thm}\label{th:prgmain}
Fix some parameters $\eps > 0$ and $R \in \N$.
Let $z$ be a standard gaussian, and let $Z = \frac{1}{\sqrt{L}}  \sum_{i=1}^{L} Y_i$, where each $Y_i$ is   $dR$-moment-matching. Then for some sufficiently large absolute constant $c$ and any polynomial $p$ of degree $d$, 
$$ \E_Z \sign(p(Z)) \geq \E_{z \sim N(0,1)^n} \sign(p(z)) -  O(\eps d^3) - L \cdot 2^{-\Omega(R)}, $$
as long as $L$ is at least $R d^c / \eps^2$.
\end{thm}

Combining the above with \cref{th:discretize} immediately implies our main result \cref{th:mainintro}.

\begin{proof}[Proof of \cref{th:mainintro}]
Given a target error $\eps'$, set $\eps = \eps'/C d^3$, and $R = C \log(d/\epsilon)$ for a sufficiently big constant so that the error in the above lemma is at most $\eps'/2$ for $L = R d^c/\eps^2 = O(d^c \log(d/\eps)/\eps^2)$. While the above theorem only gives a lower bound, we can get an upper bound by applying the result to $-p$. Now, by applying \cref{th:discretize} there exists an efficient PRG that fools degree $d$ PTFs with error at most $\eps'$ and seedlength $O(d^{O(1)} \log(n d /\eps') \log(d/\eps')/(\eps')^2$ which can be simplified to the bound in the theorem.
\end{proof}

We now prove the above theorem by the lower-sandwiching argument outlined in \cref{sec:prgoutline}.
Fix a polynomial $p(x)$ of degree $d$.
We remind the reader of our convention $\sign(t) := \ind{t \geq 0}$. 

We define the mollifier function 
$$ g(x) := \prod_{k=0}^{d-1} \rho \left( \log\left( \frac{1}{ 16 \eps^2}\frac{ \|\nabla^{k} p(x) \|^2 }{\|\nabla^{k+1} p(x) \|^2} \right) \right),
$$
where $\rho:\R \rightarrow [0,1]$ is some smooth univariate function with $\rho(t) = 0$ for $t \leq 0$, $\rho(t) = 1$ for $t \geq 1$, and $ \| \frac{\partial^k \rho}{\partial t^k} \|_{\infty} \leq k^{O(k)}$ for all $k$. \footnote{
For example, it suffices to let $\rho(t)$ be the standard mollifier $\rho(t) := 0$ for $t \leq 0$, $\rho(t) := 1$ for $t \geq 1$, and $\rho(t) := e \cdot \textnormal{exp}\left(\frac{1}{(t-1)^2 - 1} \right)$ for $t \in (0,1)$.
} 
\begin{proof}[Proof of \cref{th:prgmain}]
For every point $x \in \R^n$ we have
\begin{align*}
    \sign(p(x)) \geq \sign(p(x)) g(x).
\end{align*}
Furthermore, under the truly-random gaussian inputs $z \sim N(0,1)^n$ we have
\begin{align*}
    \E_z \sign(p(z)) g(z) \geq \E_z \sign(p(z)) - \E_z |g(z) - 1| \geq \E_z \sign(p(z))  - O(\eps d^3),
\end{align*}
where the final inequality here follows from \cref{lm:derivatives}.
Combining these, we get that
\begin{align*}
    \E_Z \sign(p(Z)) \geq \E_{z \sim N(0,1)^n} \sign(p(z))  - O( \eps d^3) - |\E_Z \sign(p(Z)) g(Z) - \E_{z \sim N(0,1)^n} \sign(p(z)) g(z)|.
\end{align*}

Thus, it suffices to bound $|\E_Z \sign(p(Z)) g(Z) - \E_{z \sim N(0,1)^n} \sign(p(z)) g(z)|$, which we do by a hybrid argument. We first represent $z$ as $z := \frac{1}{\sqrt{L}}\sum_{i=1}^L y_i$ where each $y_i$ is an independent standard gaussian.
We can replace each $Y_i$ with $y_i$ and get 
$$ |\E_Z \sign(p(Z)) g(Z) - \E_y \sign(p(y)) g(y)| \leq 2^{-\Omega(R)} L,$$ 
as a consequence of the following lemma (restatement of \cref{lm:error1}) that we prove in the next section. \cref{th:prgmain} now follows.
\end{proof}
\begin{lem}[Main hybrid-step] \label{lm:error2}
There exists a constant $c$ such that the following holds for $\lambda \leq \epsilon^2/R d^c$. For any fixed vector $x \in \R^n$, $Y$ a $dR$-moment-matching gaussian vector, and $y \sim N(0,1)^n$, 
$$ |\E_Y \sign(p(x + \sqrt{\lambda} Y)) g(x + \sqrt{\lambda} Y) - \E_y \sign(p(x + \sqrt{\lambda} y)) g(x + \sqrt{\lambda} y) | \leq \gamma = 2^{-\Omega(R)}.$$
\end{lem}
\ignore{
\begin{lem} \label{lm:error2}
Let $x \in \R^n$ be any fixed vector, and let $Y$ be a $d R$-moment-matching gaussian vector, and let $y$ be a standard gaussian vector. Then
$$ |\E_Y \sign(p(x + \sqrt{\lambda} Y)) g(x + \sqrt{\lambda} Y) - \E_y \sign(p(x + \sqrt{\lambda} y)) g(x + \sqrt{\lambda} y) | \leq 2^{-\Omega(R)},$$
as long as
$$ \lambda \leq \frac{\eps^2 }{R d^{c}} $$
for some sufficiently large absolute constant $c$.
\end{lem}}


\subsection{Analysis of the main hybrid-step}
The proof of \cref{lm:error2} is by a case-analysis as outlined in the introduction. Consider the setting as in the lemma and define 
$$\phi(z) := U_{\sqrt{1-\lambda}} p\left( \frac{z}{\sqrt{1-\lambda}}\right).$$

The core argument will be a case-analysis on the derivatives of $\phi$ at the fixed point $x$ and whether these are slow-growing. Note that if $p$ were multi-linear, then we would simply have $\phi \equiv p$. The starting point is the following re-scaling of \cref{lm:hermite1}:
\begin{equation}\label{eq:hermite1-rescaling}
 p\left( x + \sqrt{\lambda} y  \right) = \sum_{|\alpha| \leq d} \frac{\partial^{\alpha} \phi(x)}{\sqrt{\alpha !}} \lambda^{|\alpha|/2} h_{\alpha}(y) .
\end{equation}
Further, by a re-scaling of \cref{lm:hermite2}, we get the following identity which gives a nice nearly self-referential expression relating the derivatives of $p$ to those of $\phi$: 

\begin{equation}\label{eq:hermite2-rescaling}
    \partial^{\alpha}p\left( x + \sqrt{\lambda} y  \right) = \sum_{\beta \geq \alpha} \sqrt{ \frac{\alpha !}{\beta!}}  \partial^{\beta} \phi(x) \lambda^{|\beta - \alpha|/2} h_{\beta - \alpha}(y).
\end{equation}

Now, note that for a truly random gaussian $y$ we have $\partial^{\alpha} \phi(x) = \E_y \partial^\alpha p(x + \sqrt{\lambda} y) $. Thus, it is reasonable to expect that for typical points $x$ and small enough $\lambda$, $\partial^\alpha p(x + \sqrt{\lambda} y)$ will be strongly concentrated around $\partial^{\alpha}\phi(x)$. The following lemma gives quantitative bounds on how much the derivatives $\partial^\alpha p(x + \sqrt{\lambda} y)$ deviate from their expectations $\partial^\alpha \phi(x)$ for a random $y \sim N(0,1)^n$. As we will need such bounds even for $k$-moment-matching $Y$, we state the deviation bound in terms of moments: 
\begin{lem} \label{lm:moments}
Suppose $f$ is a degree-$d$ polynomial, and let $\phi(z) = U_{\sqrt{1-\lambda}} f(\frac{z}{\sqrt{1-\lambda}})$.
Consider the polynomial
$$ D(y) := \|\nabla^{k} f(x + \sqrt{\lambda} y) - \nabla^k \phi(x) \|^2 ,$$
which measures the euclidean distance between the $k$-th order derivatives $\nabla^{k} f(x + \sqrt{\lambda} y)$ and their expectations $ \nabla^k \phi(x) $.

For $y \sim N(0,1)^n$, we have the moment bound
$$ \| D(y) \|_{q/2} \leq \sum_{t =  k + 1}^d (\lambda d q)^{t - k} \| \nabla^{t} \phi(x) \|^2 .$$
That is,
$$ \left(\E_{y \sim N(0,1)^n} \|\nabla^{k} f(x + \sqrt{\lambda} y) - \nabla^k \phi(x) \|^q \right)^{1/q} \leq 
\sqrt{ \sum_{t =  k + 1}^d (\lambda d q)^{t - k} \| \nabla^{t} \phi(x) \|^2} .
$$
\end{lem}

\begin{proof}
We express
$$ D(y) = \sum_{\alpha} \left( \partial^{\alpha} f(x + \sqrt{\lambda} y) - \partial^{\alpha} \phi(x)  \right)^2
= \sum_{\alpha} \left( \sum_{\beta > \alpha} \sqrt{\frac{\alpha!}{\beta!}}\partial^{\beta} \phi(x) \lambda^{|\beta-\alpha|/2} h_{\beta - \alpha}(y) \right)^2.
$$
First, by triangle-inequality, we get
\begin{align*}
    \| D(y) \|_{q/2} &\leq   \sum_{\alpha}  \left\| \left( \sum_{\beta > \alpha} \sqrt{\frac{\alpha!}{\beta!}}\partial^{\beta} \phi(x) \lambda^{|\beta-\alpha|/2} h_{\beta - \alpha}(y) \right)^2 \right\|_{q/2} \\
    &=   \sum_{\alpha} \left\| \sum_{\beta > \alpha} \sqrt{\frac{\alpha!}{\beta!}}\partial^{\beta} \phi(x) \lambda^{|\beta-\alpha|/2} h_{\beta - \alpha}(y)  \right\|_{q}.
\end{align*}
Applying hypercontractivity, we now get 
\begin{align*}
    \| D(y) \|_{q/2}     &\leq \sum_{\alpha} \left\| U_{\sqrt{q}} \sum_{\beta > \alpha} \sqrt{\frac{\alpha!}{\beta!}}\partial^{\beta} \phi(x) \lambda^{|\beta-\alpha|/2} h_{\beta - \alpha}(y)  \right\|_{2} \\
    &= \sum_{\alpha} \sum_{\beta > \alpha} \frac{\alpha!}{\beta!} \partial^{\beta} \phi(x)^2 \lambda^{|\beta-\alpha|} q^{|\beta-\alpha|} \\
    &\leq \sum_{\alpha} \sum_{\beta > \alpha} \partial^{\beta} \phi(x)^2 \lambda^{|\beta-\alpha|} q^{|\beta-\alpha|} \\
    &= \sum_{t = k + 1}^d \binom{t}{t-k} (\lambda q)^{t - k} \|\nabla^t \phi(x) \|^2 \\
    &\leq \sum_{t = k + 1}^d (\lambda  d q)^{t - k} \|\nabla^t \phi(x) \|^2. \qedhere
\end{align*}
\end{proof}

We are now ready to prove \cref{lm:error2}. 

\begin{proof}[Proof of \cref{lm:error2}]
We study two cases:
\begin{enumerate}
    \item $x$ is poorly-behaved for $\phi$. In this case, we will show that $g(x + \sqrt{\lambda} Y) = 0$ with probability at least $1 - 2^{-\Omega(R)}$. 
    \item $x$ is well-behaved for $\phi$: In this case, we will exploit the fact that $\sign(p(x+\sqrt{\lambda Y}))$ will equal $\sign(\phi(x))$ with probability $1-2^{-\Omega(R)}$. We then have to show that $Y$ fools the mollifier $g$ which is a bit technically involved (hence we deal with this case second unlike in \cref{sec:prgoutline}). 
\end{enumerate}
We begin with the first case.

\noindent{\bf Case 1: $x$ is poorly-behaved for $\phi$}. 
Consider the case where the inequality $ \| \nabla^{k}\phi(x) \| \geq \eps \| \nabla^{k+1}\phi(x) \| $ is violated for some $k$, and indeed let $k_0$ be the largest $k$ such that this inequality is violated.
We will argue that with probability at least $1- 2^{-\Omega(R)}$, over random choice of $Y$, that
$$\| \nabla^{k_0}p(x+\sqrt{\lambda} Y) \| \leq 4 \eps \| \nabla^{k_0+1}p(x + \sqrt{\lambda} Y) \|, $$
in which case $g(x + \sqrt{\lambda} Y) = 0$. 

More specifically, we will show that it is highly likely that both
\begin{itemize}
    \item $ \| \nabla^{k_0}p(x+\sqrt{\lambda} Y) \| \leq 2 \eps \| \nabla^{k_0+1}\phi(x) \| $, and
    \item $ \| \nabla^{k_0+1}p(x+\sqrt{\lambda} Y) \| \geq \frac{1}{2}  \| \nabla^{k_0+1}\phi(x) \| $.
\end{itemize}

For this, we will use \cref{eq:hermite2-rescaling} and \cref{lm:moments}. Supposing $k_0$ is the largest $k$ such that
$$ \| \nabla^{k}\phi(x) \| < \eps \| \nabla^{k+1}\phi(x) \|, $$
we have
\begin{itemize}
    \item  $ \| \nabla^{k_0}\phi(x) \| \leq \eps  \| \nabla^{k_0+1}\phi(x) \| $ and 
    \item $  \| \nabla^{k_0+1}\phi(x) \| \geq \eps^{t}  \| \nabla^{k_0+1+t}\phi(x) \|  $ for all $t \geq 0$.
\end{itemize}

\cref{lm:moments} therefore gives the bounds
$$ \left(\E_Y \|\nabla^{k_0} p(x + \sqrt{\lambda} Y) - \nabla^{k_0} \phi(x) \|^R \right)^{1/R} 
\leq \eps \| \nabla^{k_0+1}\phi(x) \| \sqrt{\sum_{t \geq 1} (\lambda d R / \eps^2 )^{t} } 
$$ 
and 
$$ \left(\E_Y \|\nabla^{k_0+1} p(x + \sqrt{\lambda} Y) - \nabla^{k_0+1} \phi(x) \|^R \right)^{1/R} 
\leq  \| \nabla^{k_0+1}\phi(x) \| \sqrt{\sum_{t \geq 1} (\lambda d R / \eps^2 )^{t} } .
$$ 
So, as long as $\lambda d R / \eps^2$ is at most a sufficiently small constant, we conclude that the following bounds hold with probability at least $1 - 2^{-R}$:

\begin{itemize}
    \item $ \| \nabla^{k_0}p(x+\sqrt{\lambda} Y) \| \leq  
    \| \nabla^{k_0}\phi(x) \| + \|\nabla^{k_0} p(x + \sqrt{\lambda} Y) - \nabla^{k_0} \phi(x) \| 
    \leq 2 \eps \| \nabla^{k_0+1}\phi(x) \| $, and
    \item $ \| \nabla^{k_0+1}p(x+\sqrt{\lambda} Y) \| 
    \geq \| \nabla^{k_0+1}\phi(x) \| - \|\nabla^{k_0+1} p(x + \sqrt{\lambda} Y) - \nabla^{k_0+1} \phi(x) \|
    \geq \frac{1}{2}  \| \nabla^{k_0+1}\phi(x) \| $.
\end{itemize}
In the case that these bounds hold, we get $$\| \nabla^{k_0}p(x+\sqrt{\lambda} Y) \| \leq 4 \eps \| \nabla^{k_0+1}p(x + \sqrt{\lambda} Y) \|,$$ and so
$ g(x + \sqrt{\lambda} Y) = 0 .$ As this holds with probability at least $1 - 2^{-\Omega(R)}$ for both $y \sim N(0,1)^n$ as well as $Y$, the conclusion of \cref{lm:error2} follows. This finishes the proof of Case 1. $\qedsymbol$

\noindent{\bf Case 2: $x$ is well-behaved for $\phi$.}  We now consider the complimentary case where
$$ \| \nabla^k \phi(x) \| \geq \eps \| \nabla^{k+1} \phi(x) \| $$
for all $k = 0,1,\ldots,d-1$. Consider the normalized polynomial
$$ f(y) := \frac{p(x + \sqrt{\lambda} y)}{\phi(x)} = 1 + \frac{1}{\phi(x)} \sum_{\alpha \neq 0} \partial^{\alpha} \phi(x) \lambda^{|\alpha|/2} h_{\alpha}(y) .$$
Using hypercontractivity, we bound the $R$-th moment of $f(y) - 1$ by its $\sqrt{R}$-hypervariance:
$$ \|f(y) - 1\|_{R} \leq \|U_{\sqrt{R}} \left(f(y) - 1\right)\|_{2} \leq  \sqrt{ \sum_{k \geq 1}  \left( \frac{\lambda R}{\eps^2}  \right)^k } \leq \frac{1}{2} .
$$
So, by a Markov argument, we have
$$ \prob{\sign(p(x+\sqrt{\lambda}Y)) \neq \sign(\phi(x))} \leq 2^{-R} ,$$
and this holds whenever $Y$ is $k$-moment-matching for $k \geq d R$. 
So, $\sign(p(x+\sqrt{\lambda}Y))$ is nearly a constant for random $Y$; it remains to show that $Y$ fools $g(x + \sqrt{\lambda} Y)$. We do this by (essentially) truncating the Taylor-series of $g$ about $x$ so that we are left with a degree $dR$ polynomial, which is fooled by $Y$. The truncation-error will be small because our assumption, 
$$\| \nabla^{k}\phi(x) \| \geq \eps \| \nabla^{k+1}\phi(x) \| \textnormal{ for all } k, $$ 
gives us good control on the $R$-th order moments of the deviations  $\| \nabla^k \phi(x) \| - \| \nabla^k p(x + \sqrt{\lambda} Y) \|$. The exact calculations are somewhat cumbersome and are given below. 
We will show that $Y$ fools the mollifier function
$$ g(x + \sqrt{\lambda} y) = \prod_{k=0}^{d-1} \rho \left( \log\left( \frac{1}{ 16 \eps^2}\frac{ \|\nabla^{k} p(x + \sqrt{\lambda} y) \|^2 }{\|\nabla^{k+1} p(x + \sqrt{\lambda} y) \|^2} \right) \right).
$$
To simplify notation we define the shifted function $\sigma(t) := \rho(t - \log(16 \eps^2))$, and express
$$ g(x + \sqrt{\lambda} y) = \prod_{k=0}^{d-1} \sigma \left( \log \|\nabla^{k} p(x + \sqrt{\lambda} y) \|^2 - \log  \|\nabla^{k+1} p(x + \sqrt{\lambda} y) \|^2 \right).
$$
It will be convenient to think of $g$ (redundantly) as function of $2d$ auxiliary variables $s_1 \ldots s_d$, $t_1, \ldots t_d$, which we will eventually fix to
\begin{itemize}
    \item $s_i := \|\nabla^{i-1} p(x + \sqrt{\lambda} y) \|^2$
    \item $t_i := \|\nabla^{i} p(x + \sqrt{\lambda} y) \|^2$,
\end{itemize}
so we write
$$ g(s,t) := \prod_{i=1}^d \sigma \left( \log(s_i) - \log(t_i) \right).
$$
We Taylor-expand $g(s,t)$ around the points
\begin{itemize}
    \item $a_i := \|\nabla^{i-1} \phi(x) \|^2$
    \item $b_i := \|\nabla^{i} \phi(x) \|^2$,
\end{itemize}
which gives
$$ g(s,t) = \ell(s,t) + h(s,t), $$
with low-degree part
$$\ell(s,t) := \sum_{\substack{\alpha, \beta \in \N^d \\ |\alpha| + |\beta| < R}} 
\frac{\partial^{\alpha}_s \partial^{\beta}_t g(a,b)}{\alpha ! \beta !} \left( s - a \right)^{\alpha} \left(t - b \right)^{\beta}
$$
and remainder 
$$ |h(s,t)| \leq  \sum_{\substack{\alpha, \beta \in \N^d \\ |\alpha| + |\beta| = R}} 
 \frac{ |\partial^{\alpha}_s \partial^{\beta}_t g(s^*,t^*) |}{\alpha ! \beta !}  \left| s - a \right|^{\alpha} \left| t - b \right|^{\beta},
$$
where ``$ |\partial^{\alpha}_s \partial^{\beta}_t g(s^*,t^*) |$'' is notation for the maximum magnitude of $\partial^{\alpha}_s \partial^{\beta}_t g$ on any point on the line segment from $(a,b)$ to $(s,t)$. 
We need the following fact to bound the size of the derivatives of $g$,

\begin{claim}
Suppose $\sigma$ is a smooth univariate function with uniform derivative bounds
$$ \| \sigma^{(n)} \|_{\infty} \leq n^{O(n)} .
$$
The bivariate function 
$$ r(u,v) :=  \sigma( \log(u) - \log(v)) $$
has derivatives bounded in size by
$$ \left|  \frac{\partial^{n}}{\partial u^{n}} \frac{\partial^{m}}{\partial v^m} r(u,v) \right| \leq 
\frac{n^{O(n)}}{|u|^{n}} \frac{m^{O(m)}}{|v|^{m}}.
$$
\end{claim}
This claim follows easily from the generalized chain rule (Faà di Bruno's formula).
As a result, we get the derivative bounds
$$ \left| \partial^{\alpha}_s \partial^{\beta}_{t} g(s,t) \right| \leq 
\frac{|\alpha|^{O(|\alpha|)}}{|s^\alpha|} \frac{|\beta|^{O(|\beta|)}}{|t^\beta|}.
$$
Using this, we bound the remainder
$$ |h(s,t)| \leq \sum_{\substack{\alpha, \beta \in \N^d \\ |\alpha| + |\beta| = R}} 
d^{O(R)}
\prod_{i = 1}^d \left( \frac{|1 - \tfrac{s_i}{a_i}|}{1 - |1 - \tfrac{s_i}{a_i}|} \right)^{\alpha_i }
 \left( \frac{|1 - \tfrac{t_i}{b_i}|}{1 - |1 - \tfrac{t_i}{b_i}|} \right)^{\beta_i}.
$$
Now, consider the event $E$ (which depends on $y$) that 
$$  (1-\delta) \|\nabla^{i} \phi(x) \|^2 \leq \|\nabla^{i} p(x + \sqrt{\lambda} y) \|^2 \leq (1+\delta) \|\nabla^{i} \phi(x) \|^2
$$
holds for all $i$, where $\delta \leq 1/2$ is a parameter we will set shortly.
In the case that this indeed holds, we get
$$ |h(s,t)| \leq d^{O(R)} O(\delta)^{R} .$$
We set $\delta$ just small enough to ensure
$$ |h(s,t)| \leq 2^{-R} .$$
Now, we express $g$ (which we now think of as a function of the underlying variable $y$) as
\begin{align*}
    g &= g \cdot \1_E + g \cdot \1_{\bar{E}} \\
    &= \ell \cdot \1_E + h \cdot \1_E + g \cdot \1_{\bar{E}} \\
    &= \ell - \ell \cdot \1_{\bar{E}} + h \cdot \1_E + g  \cdot \1_{\bar{E}} ,
\end{align*}
and we obtain the pointwise bound
$$ \left| g - \ell \right| \leq 2^{-R} + \1_{\bar{E}} + |\ell| \cdot \1_{\bar{E}} .$$
On average over $Y$, we get truncation error
\begin{align*} \E_Y \left| g(x + \sqrt{\lambda} Y) - \ell(Y) \right| &\leq 2^{-R}  + \E_Y  \1_{\bar{E}}(Y) + \sqrt{\E_Y \ell^2(Y)} \sqrt{\E_Y \1_{\bar{E}}(Y)} \\
&\leq 2^{-R} +  O\left( \frac{d}{\delta} \right)^R \cdot \left( \frac{\lambda d R}{\eps^2} \right)^{-\Omega(R)} \\
&\leq 2^{-R}  + d^{O(1)} \cdot \left( \frac{\lambda d R}{\eps^2} \right)^{-\Omega(R)} 
\end{align*}
where the second inequality here follows from the moment bounds in \cref{lm:moments}. 
As required by the conditions of \cref{lm:error2}, we insist that $\lambda$ is small enough that this error is at most $2^{-\Omega(R)}$. Since this bound holds also for truly-random standard gaussian $y$, and $\E_Y \ell(Y) = \E_y \ell(y)$, we obtain the desired bound
\begin{equation*}
|\E_Y g(x + \sqrt{\lambda} Y) - \E_y g(x + \sqrt{\lambda} y)| \leq 2^{-\Omega(R)} . 
\end{equation*}

This finishes the proof in Case 2 and hence of \cref{lm:error2}. \qedhere

\end{proof}

\ignore{
\begin{proof}
We argue by a case analysis based on the behavior of the polynomial $p$ at the fixed point $x$.
For this, we use the following re-scaling of \cref{lm:hermite1}:
\begin{lem}[Re-scaling of \cref{lm:hermite1}] \label{lm:hermite1-rescaling}
Let $p$ be a degree-$d$ polynomial. Then
$$ p\left( x + \sqrt{\lambda} y  \right) = \sum_{|\alpha| \leq d} \frac{\partial^{\alpha} \phi(x)}{\sqrt{\alpha !}} \lambda^{|\alpha|/2} h_{\alpha}(y) ,$$
where $\phi(x) := U_{\sqrt{1-\lambda}} p\left( \frac{x}{\sqrt{1-\lambda}}\right).$
\end{lem}

We do our case analysis based on the derivatives of $\phi(x)$. 
To get a handle on this polynomial, note that if $p(x)$ is multilinear then we simply have $\phi(x) = p(x)$. Furthermore, note that for a truly random gaussian $y$ we have $\partial^{\alpha} \phi(x) = \E_y \partial^\alpha p(x + \sqrt{\lambda} y) $. Thus, it is reasonable to expect that for typical points $x$ and small enough $\lambda$, $\partial^\alpha p(x + \sqrt{\lambda} y)$ will be strongly concentrated around $\phi(x)$.

\paragraph{Case 1: $\|\nabla^{k}p\|$ is well-behaved for all $k$.} We begin with the first case. Suppose that
$$ \| \nabla^k \phi(x) \| \geq \eps \| \nabla^{k+1} \phi(x) \| $$
for all $k = 0,1,\ldots,d-1$. Consider the normalized polynomial

$$ f(y) := \frac{p(x + \sqrt{\lambda} y)}{\phi(x)} = 1 + \frac{1}{\phi(x)} \sum_{\alpha \neq 0} \partial^{\alpha} \phi(x) \lambda^{|\alpha|/2} h_{\alpha}(y) .$$
Using hypercontractivity, we bound the $R$-th moment of $f(y) - 1$ by its $\sqrt{R}$-hypervariance:
$$ \|f(y) - 1\|_{R} \leq \|U_{\sqrt{R}} \left(f(y) - 1\right)\|_{2} \leq  \sqrt{ \sum_{k \geq 1}  \left( \frac{\lambda R}{\eps^2}  \right)^k } \leq \frac{1}{2} .
$$
So, by a Markov argument, we have
$$ \prob{\sign(p(x+\sqrt{\lambda}Y)) \neq \sign(\phi(x))} \leq 2^{-R} ,$$
and this holds whenever $Y$ is $k$-moment-matching for $k \geq d R$. 
So, $\sign(p(x+\sqrt{\lambda}Y))$ is nearly a constant for random $Y$; it remains to show that $Y$ fools $g(x + \sqrt{\lambda} Y)$. We do this by (essentially) truncating the Taylor-series of $g$ about $x$ so that we are left with a degree $dR$ polynomial, which is fooled by $Y$. The truncation-error will be small because our assumption, 
$$\| \nabla^{k}\phi(x) \| \geq \eps \| \nabla^{k+1}\phi(x) \| \textnormal{ for all } k, $$ 
gives us good control on the $R$-th order moments of the deviations  $\| \nabla^k \phi(x) \| - \| \nabla^k p(x + \sqrt{\lambda} Y) \|$. The exact calculations are fairly routine, so we return to this after showing how to handle the second case. 

\paragraph{Case 2: $\|\nabla^{k_0}p\|$ is poorly-behaved.} 
Now, we consider the complement case where the inequality $ \| \nabla^{k}\phi(x) \| \geq \eps \| \nabla^{k+1}\phi(x) \| $ is violated for some $k$, and indeed let $k_0$ be the largest $k$ such that this inequality is violated.
We will argue that it is highly likely, over random choice of $Y$, that
$$\| \nabla^{k_0}p(x+\sqrt{\lambda} Y) \| \leq 4 \eps \| \nabla^{k_0+1}p(x + \sqrt{\lambda} Y) \|, $$
in which case $g(x + \sqrt{\lambda} Y) = 0$. 

More specifically, we will show that it is highly likely that both
\begin{itemize}
    \item $ \| \nabla^{k_0}p(x+\sqrt{\lambda} Y) \| \leq 2 \eps \| \nabla^{k_0+1}\phi(x) \| $, and
    \item $ \| \nabla^{k_0+1}p(x+\sqrt{\lambda} Y) \| \geq \frac{1}{2}  \| \nabla^{k_0+1}\phi(x) \| $.
\end{itemize}

For this, we use the following re-scaling of \cref{lm:hermite2}, which gives a nice, nearly self-referential expression for the higher-order derivatives.
\begin{lem}[Re-scaling of \cref{lm:hermite2}] \label{lm:hermite2-recaling}
Let $p$ be a degree-$d$ polynomial. Then
$$ \partial^{\alpha}p\left( x + \sqrt{\lambda} y  \right) = \sum_{\beta \geq \alpha} \sqrt{ \frac{\alpha !}{\beta!}}  \partial^{\beta} \phi(x) \lambda^{|\beta - \alpha|/2} h_{\beta - \alpha}(y) ,$$
where  $\phi(x) := U_{\sqrt{1-\lambda}} p\left( \frac{x}{\sqrt{1-\lambda}}\right).$
\end{lem}

With this we derive the following bound on the higher-moments of the deviation of $\nabla^{k} p(x + \sqrt{\lambda} y)$ from its expectation.
\begin{lem} \label{lm:moments}
Suppose $f$ is a degree-$d$ polynomial, and let $\phi(x) = U_{\sqrt{1-\lambda}} f(\frac{x}{\sqrt{1-\lambda}})$.
Consider the polynomial
$$ D(y) := \|\nabla^{k} f(x + \sqrt{\lambda} y) - \nabla^k \phi(x) \|^2 ,$$
which measures the euclidean distance between the $k$-th order derivatives $\nabla^{k} f(x + \sqrt{\lambda} y)$ and their expectations $ \nabla^k \phi(x) $.

For $y \sim N(0,1)^n$, we have the moment bound
$$ \| D(y) \|_{q/2} \leq \sum_{t =  k + 1}^d (\lambda d q)^{t - k} \| \nabla^{t} \phi(x) \|^2 .$$
That is,
$$ \left(\E_{y \sim N(0,1)^n} \|\nabla^{k} f(x + \sqrt{\lambda} y) - \nabla^k \phi(x) \|^q \right)^{1/q} \leq 
\sqrt{ \sum_{t =  k + 1}^d (\lambda d q)^{t - k} \| \nabla^{t} \phi(x) \|^2} .
$$
\end{lem}

\begin{proof}
We express
$$ D(y) = \sum_{\alpha} \left( \partial^{\alpha} f(x + \sqrt{\lambda} y) - \partial^{\alpha} \phi(x)  \right)^2
= \sum_{\alpha} \left( \sum_{\beta > \alpha} \sqrt{\frac{\alpha!}{\beta!}}\partial^{\beta} \phi(x) \lambda^{|\beta-\alpha|/2} h_{\beta - \alpha}(y) \right)^2.
$$
Now, using first the triangle-inequality and then hypercontractivity we get
\begin{align*}
    \| D(y) \|_{q/2} &\leq   \sum_{\alpha}  \left\| \left( \sum_{\beta > \alpha} \sqrt{\frac{\alpha!}{\beta!}}\partial^{\beta} \phi(x) \lambda^{|\beta-\alpha|/2} h_{\beta - \alpha}(y) \right)^2 \right\|_{q/2} \\
    &=   \sum_{\alpha} \left\| \sum_{\beta > \alpha} \sqrt{\frac{\alpha!}{\beta!}}\partial^{\beta} \phi(x) \lambda^{|\beta-\alpha|/2} h_{\beta - \alpha}(y)  \right\|_{q} \\
    &\leq \sum_{\alpha} \left\| U_{\sqrt{q}} \sum_{\beta > \alpha} \sqrt{\frac{\alpha!}{\beta!}}\partial^{\beta} \phi(x) \lambda^{|\beta-\alpha|/2} h_{\beta - \alpha}(y)  \right\|_{2} \\
    &= \sum_{\alpha} \sum_{\beta > \alpha} \frac{\alpha!}{\beta!} \partial^{\beta} \phi(x)^2 \lambda^{|\beta-\alpha|} q^{|\beta-\alpha|} \\
    &\leq \sum_{\alpha} \sum_{\beta > \alpha} \partial^{\beta} \phi(x)^2 \lambda^{|\beta-\alpha|} q^{|\beta-\alpha|} \\
    &= \sum_{t = k + 1}^d \binom{t}{t-k} (\lambda q)^{t - k} \|\nabla^t \phi(x) \|^2 \\
    &\leq \sum_{t = k + 1}^d (\lambda  d q)^{t - k} \|\nabla^t \phi(x) \|^2. \qedhere
\end{align*}
\end{proof}

Suppose $k_0$ is the largest $k$ such that
$$ \| \nabla^{k}\phi(x) \| < \eps \| \nabla^{k+1}\phi(x) \|. $$
We have
\begin{itemize}
    \item  $ \| \nabla^{k_0}\phi(x) \| \leq \eps  \| \nabla^{k_0+1}\phi(x) \| $ and 
    \item $  \| \nabla^{k_0+1}\phi(x) \| \geq \eps^{t}  \| \nabla^{k_0+1+t}\phi(x) \|  $ for all $t \geq 0$.
\end{itemize}

Our recent \cref{lm:moments} therefore gives the bounds
$$ \left(\E_Y \|\nabla^{k_0} p(x + \sqrt{\lambda} Y) - \nabla^{k_0} \phi(x) \|^R \right)^{1/R} 
\leq \eps \| \nabla^{k_0+1}\phi(x) \| \sqrt{\sum_{t \geq 1} (\lambda d R / \eps^2 )^{t} } 
$$ 
and 
$$ \left(\E_Y \|\nabla^{k_0+1} p(x + \sqrt{\lambda} Y) - \nabla^{k_0+1} \phi(x) \|^R \right)^{1/R} 
\leq  \| \nabla^{k_0+1}\phi(x) \| \sqrt{\sum_{t \geq 1} (\lambda d R / \eps^2 )^{t} } .
$$ 
So, as long as $\lambda d R / \eps^2$ is at most a sufficiently small constant, we conclude that the following bounds hold with probability at least $1 - 2^{-R}$:

\begin{itemize}
    \item $ \| \nabla^{k_0}p(x+\sqrt{\lambda} Y) \| \leq  
    \| \nabla^{k_0}\phi(x) \| + \|\nabla^{k_0} p(x + \sqrt{\lambda} Y) - \nabla^{k_0} \phi(x) \| 
    \leq 2 \eps \| \nabla^{k_0+1}\phi(x) \| $, and
    \item $ \| \nabla^{k_0+1}p(x+\sqrt{\lambda} Y) \| 
    \geq \| \nabla^{k_0+1}\phi(x) \| - \|\nabla^{k_0+1} p(x + \sqrt{\lambda} Y) - \nabla^{k_0+1} \phi(x) \|
    \geq \frac{1}{2}  \| \nabla^{k_0+1}\phi(x) \| $.
\end{itemize}
In the case that these bounds hold, we get $$\| \nabla^{k_0}p(x+\sqrt{\lambda} Y) \| \leq 4 \eps \| \nabla^{k_0+1}p(x + \sqrt{\lambda} Y) \|,$$ and so
$ g(x + \sqrt{\lambda} Y) = 0 .$

\textbf{Truncation error for Case 1.} Now, we return to show that $Y$ fools the mollifier function
$$ g(x + \sqrt{\lambda} y) = \prod_{k=0}^{d-1} \rho \left( \log\left( \frac{1}{ 16 \eps^2}\frac{ \|\nabla^{k} p(x + \sqrt{\lambda} y) \|^2 }{\|\nabla^{k+1} p(x + \sqrt{\lambda} y) \|^2} \right) \right).
$$
To simplify notation we define the shifted function $\sigma(t) := \rho(t - \log(16 \eps^2))$, and express
$$ g(x + \sqrt{\lambda} y) = \prod_{k=0}^{d-1} \sigma \left( \log \|\nabla^{k} p(x + \sqrt{\lambda} y) \|^2 - \log  \|\nabla^{k+1} p(x + \sqrt{\lambda} y) \|^2 \right).
$$
It will be convenient to think of $g$ (redundantly) as function of $2d$ auxiliary variables $s_1 \ldots s_d$, $t_1, \ldots t_d$, which we will eventually fix to
\begin{itemize}
    \item $s_i := \|\nabla^{i-1} p(x + \sqrt{\lambda} y) \|^2$
    \item $t_i := \|\nabla^{i} p(x + \sqrt{\lambda} y) \|^2$,
\end{itemize}
so we write
$$ g(s,t) := \prod_{i=1}^d \sigma \left( \log(s_i) - \log(t_i) \right).
$$
We taylor-expand $g(s,t)$ around the points
\begin{itemize}
    \item $a_i := \|\nabla^{i-1} \phi(x) \|^2$
    \item $b_i := \|\nabla^{i} \phi(x) \|^2$,
\end{itemize}
which gives
$$ g(s,t) = \ell(s,t) + h(s,t), $$
with low-degree part
$$\ell(s,t) := \sum_{\substack{\alpha, \beta \in \N^d \\ |\alpha| + |\beta| < R}} 
\frac{\partial^{\alpha}_s \partial^{\beta}_t g(a,b)}{\alpha ! \beta !} \left( s - a \right)^{\alpha} \left(t - b \right)^{\beta}
$$
and remainder 
$$ |h(s,t)| \leq  \sum_{\substack{\alpha, \beta \in \N^d \\ |\alpha| + |\beta| = R}} 
 \frac{ |\partial^{\alpha}_s \partial^{\beta}_t g(s^*,t^*) |}{\alpha ! \beta !}  \left| s - a \right|^{\alpha} \left| t - b \right|^{\beta},
$$
where ``$ |\partial^{\alpha}_s \partial^{\beta}_t g(s^*,t^*) |$'' is notation for the maximum magnitude of $\partial^{\alpha}_s \partial^{\beta}_t g$ on any point on the line segment from $(a,b)$ to $(s,t)$. 
We need the following fact to bound the size of the derivatives of $g$,

\begin{claim}
Suppose $\sigma$ is a smooth univariate function with uniform derivative bounds
$$ \| \sigma^{(n)} \|_{\infty} \leq n^{O(n)} .
$$
The bivariate function 
$$ r(u,v) :=  \sigma( \log(u) - \log(v)) $$
has derivatives bounded in size by
$$ \left|  \frac{\partial^{n}}{\partial u^{n}} \frac{\partial^{m}}{\partial v^m} r(u,v) \right| \leq 
\frac{n^{O(n)}}{|u|^{n}} \frac{m^{O(m)}}{|v|^{m}}.
$$
\end{claim}
This claim follows easily from the generalized chain rule (Faà di Bruno's formula).
As a result, we get the derivative bounds
$$ \left| \partial^{\alpha}_s \partial^{\beta}_{t} g(s,t) \right| \leq 
\frac{|\alpha|^{O(|\alpha|)}}{|s^\alpha|} \frac{|\beta|^{O(|\beta|)}}{|t^\beta|}.
$$
Using this, we bound the remainder
$$ |h(s,t)| \leq \sum_{\substack{\alpha, \beta \in \N^d \\ |\alpha| + |\beta| = R}} 
d^{O(R)}
\prod_{i = 1}^d \left( \frac{|1 - \tfrac{s_i}{a_i}|}{1 - |1 - \tfrac{s_i}{a_i}|} \right)^{\alpha_i }
 \left( \frac{|1 - \tfrac{t_i}{b_i}|}{1 - |1 - \tfrac{t_i}{b_i}|} \right)^{\beta_i}.
$$
Now, consider the event $E$ (which depends on $y$) that 
$$  (1-\delta) \|\nabla^{i-1} \phi(x) \|^2 \leq \|\nabla^{i-1} p(x + \sqrt{\lambda} y) \|^2 \leq (1+\delta) \|\nabla^{i-1} \phi(x) \|^2
$$
holds for all $k$, where $\delta \leq 1/2$ is a parameter we will set shortly.
In the case that this indeed holds, we get
$$ |h(s,t)| \leq d^{O(R)} O(\delta)^{R} .$$
We set $\delta$ just small enough to ensure
$$ |h(s,t)| \leq 2^{-R} .$$
Now, we express $g$ (which we now think of as a function of the underlying variable $y$) as
\begin{align*}
    g &= g \cdot \1_E + g \cdot \1_{\bar{E}} \\
    &= \ell \cdot \1_E + h \cdot \1_E + g \cdot \1_{\bar{E}} \\
    &= \ell - \ell \cdot \1_{\bar{E}} + h \cdot \1_E g \cdot \1_{\bar{E}} ,
\end{align*}
and we obtain the pointwise bound
$$ \left| g - \ell \right| \leq 2^{-R} + \1_{\bar{E}} + \ell \cdot \1_{\bar{E}} .$$
On average over $Y$, we get truncation error
\begin{align*} \E_Y \left| g(x + \sqrt{\lambda} Y) - \ell(Y) \right| &\leq 2^{-R}  + \E_Y  \1_{\bar{E}}(Y) + \sqrt{\E_Y \ell^2(Y)} \sqrt{\E_Y \1_{\bar{E}}(Y)} \\
&\leq 2^{-R} +  O\left( \frac{d}{\delta} \right)^R \cdot \left( \frac{\lambda d R}{\eps^2} \right)^{-\Omega(R)} \\
&\leq 2^{-R}  + d^{O(1)} \cdot \left( \frac{\lambda d R}{\eps^2} \right)^{-\Omega(R)} 
\end{align*}
where the second inequality here follows from the moment bounds in \cref{lm:moments}. 
As required by the conditions of \cref{lm:error2}, we insist that $\lambda$ is small enough that this error is at most $2^{-\Omega(R)}$. Since this bound holds also for truly-random standard gaussian $y$, and $\E_Y \ell(Y) = \E_y \ell(y)$, we obtain the desired bound
\begin{equation*}
|\E_Y g(x + \sqrt{\lambda} Y) - \E_y g(x + \sqrt{\lambda} y)| \leq 2^{-\Omega(R)} . \qedhere
\end{equation*}
\end{proof}}

\bibliographystyle{alpha}
\bibliography{gaussianptf}

\end{document}